\documentclass[a4paper,UKenglish,cleveref, autoref, thm-restate]{lipics-v2021}

\hideLIPIcs

\title{Strong (D)QBF Dependency Schemes via Pure Paths with Applications to Proof Checking} 

\author{Leroy Chew}{Czech Technical University in Prague, Prague, Czech Republic\and \url{https://leroychew.wordpress.com/} }{lchew@ac.tuwien.ac.at}{https://orcid.org/0000-0003-0226-2832}{This project is supported by FWF ESPRIT grant number ESP-197, by the European Union under the project ROBOPROX (reg. no. CZ.02.01.01/00/22\_008/0004590)
	and by the Czech Science Foundation project 24-12759S }
\author{Tom\'a\v{s} Peitl}{TU Wien, Vienna, Austria \and \url{https://ac.tuwien.ac.at/people/peitl}}{peitl@ac.tuwien.ac.at}{https://orcid.org/0000-0001-7799-1568}{}
\authorrunning{Leroy Chew, Tom\'a\v{s} Peitl}
\Copyright{Leroy Nicholas Chew, Tom\'a\v{s} Peitl}
\ccsdesc[100]{Theory of computation~Proof Complexity}

\keywords{DQBF, QBF, Qute, Proof Systems, Dependency Schemes, Dependency Learning, Skolem functions} 

\category{} 

\relatedversion{} 

\supplement{\url{https://doi.org/10.5281/zenodo.20203771}}

\acknowledgements{Thanks to Martina Seidl for discussions on this topic, to Mirek Olsak and Zarathustra Goertzel for help debugging DQRAT-check and to the anonymous reviewers for improvements to the paper.}

\nolinenumbers 

\titlerunning{Strong (D)QBF Dependency Schemes via Pure Paths}

\usepackage{hyperref}
\usepackage{todonotes}

\usepackage{listings}

\usepackage{multicol}
\RequirePackage{amssymb,amsmath}
\RequirePackage{xspace,xcolor}
\usepackage{url}
\usepackage{algpseudocode}
\usepackage{xifthen}

\newlength{\breite}
\usepackage{bm}

\definecolor{violet}{RGB}{138,43,226}
\definecolor{forestgreen}{RGB}{34,139,34}
\definecolor{darkblue}{RGB}{102,0,204}
\definecolor{pink}{RGB}{255,192,20}
\definecolor{gold}{RGB}{255,215,0}

\newcommand{\Math}[1]{\ensuremath{#1}\xspace}

\DeclareMathOperator*{\dom}{\operatorname{dom}}

\newcommand{\ComplexityClassFont}[1]{\mathsf{#1}}

\newcommand{\NP}{\ComplexityClassFont{NP}}
\newcommand{\CoNP}{\ComplexityClassFont{CoNP}}
\newcommand{\PSPACE}{\ComplexityClassFont{PSPACE}}
\newcommand{\NEXP}{\ComplexityClassFont{NEXPTIME}}

\DeclareMathOperator*{\var}{\operatorname{var}}

\RequirePackage{bussproofs/bussproofs}
\newcommand{\ProofSystemsFont}[1]{\mathsf{#1}}
\newcommand{\DefineProofSystem}[2]{\expandafter\def\csname#1\endcsname{\Math{\ComplexityClassFont{#2}}}}

\newcommand{\Frege}[1][]{\Math{\ifthenelse{\isempty{#1}}{\ProofSystemsFont{Frege}}{#1\text{-}\ProofSystemsFont{Frege}}}}
\newcommand{\Res}[1][]{\Math{\ifthenelse{\isempty{#1}}{\ProofSystemsFont{Res}}{#1\text{-}\ProofSystemsFont{Res}}}}
\DefineProofSystem{eFrege}{e\Frege[]}
\newcommand{\drat}{\textsf{DRAT}\xspace}

\newcommand{\Red}{\Math{\boldsymbol{\forall}\kern .04ex\ProofSystemsFont{red}}}
\newcommand{\red}{\Math{\,\raisebox{.2ex}{$\scriptstyle+$}\,\Red}}

\usepackage{mathtools}

\newcommand{\D}[1]{\ensuremath{\mathrm{D}^{\Pi}_{{\scriptstyle #1}}}\xspace}
\newcommand{\Dz}[1]{\ensuremath{\mathcal{D}^{\mbox{\scriptsize \upshape #1}}}\xspace}

\newcommand{\Drrs}{\Dz{rrs}}
\newcommand{\Dstd}{\Dz{std}}
\newcommand{\Dtrv}{\Dz{trv}}

\newcommand{\Dpu}{\Dz{$\forall$pure}}
\newcommand{\pathc}[1]{\ensuremath{\mathfrak{C}^{\mbox{\scriptsize \upshape #1}}}\xspace}
\newcommand{\pathl}[1]{\ensuremath{\mathfrak{L}^{\mbox{\scriptsize \upshape #1}}}\xspace}
\newcommand{\pathcpu}{\pathc{$\forall$pure}}
\newcommand{\pathlpu}{\pathl{$\forall$pure}}
\newcommand{\lsim}[1]{\ensuremath{\rightarrow_{\mbox{\scriptsize \upshape #1}}}}
\newcommand{\rrssim}{\lsim{rrs}}
\newcommand{\nrrssim}{\ensuremath{\not\rightarrow_{\mbox{\scriptsize \upshape rrs}}}}
\newcommand{\pusim}{\ensuremath{\rightarrow_{\mbox{\scriptsize \upshape $\forall$pure}}}}
\newcommand{\npusim}{\ensuremath{\not\rightarrow_{\mbox{\scriptsize \upshape $\forall$pure}}}}
\newcommand{\lplr}{\textsf{Loc}$^{\forall \textrm{pure}}$-Red\xspace}

\DefineProofSystem{Gfull}{G}
\newcommand{\FregeRed}[1][]{\Frege\!\!\red}
\newcommand{\eFregeRed}[1][]{\eFrege\!\!\red}
\newcommand{\qrc}{\textsf{Q-Res}\xspace}

\newcommand{\qurc}{\textsf{QU-Res}\xspace}
\newcommand{\ecalculus}{$\forall$\textsf{Exp+Res}\xspace}
\newcommand{\irc}{\textsf{IR-calc}\xspace}
\newcommand{\irmc}{\textsf{IRM-calc}\xspace}
\newcommand{\idrc}{\textsf{IR(\Drrs)-calc}\xspace}
\newcommand{\lqrc}{\textsf{LD-Q-Res}\xspace}

\newcommand{\lqurc}{\textsf{LQU}\textsf{-\hspace{1pt}Res}\xspace}
\newcommand{\lquprc}{\textsf{LQU}$^+$\textsf{-\hspace{1pt}Res}\xspace}
\newcommand{\qdrc}{\textsf{Q(\Drrs)-Res}\xspace}
\newcommand{\qsrc}{\textsf{Q(\Dstd)-Res}\xspace}
\newcommand{\lqDrc}{\textsf{LD-Q(\ensuremath{\mathcal{D}})-Res}\xspace}
\newcommand{\lqdrc}{\textsf{LD-Q(\Drrs)-Res}\xspace}
\newcommand{\lqdpurc}{\textsf{LD-Q(\Dpu)-Res}\xspace}
\newcommand{\qrat}{\textsf{QRAT}\xspace}
\newcommand{\qrp}{\textsf{QRP}\xspace}

\DeclareMathOperator*{\xor}{\operatorname{xor}}

\def\qparity{\textsc{QParity}}

\def\lqparity{\textup{\textsc{LQParity}}}

\DefineProofSystem{dFrege}{IndExt\Frege[]}
\DefineProofSystem{dRes}{IndExt QU\Res[]}
\newcommand{\dFregeRed}[1][]{\dFrege\!\!\red}

\newcommand{\dqrat}{\textsf{DQRAT}\xspace}
\newcommand{\dqratpu}{\textsf{DQRAT}+\Dpu\xspace}

\newcommand{\OC}[2]{\Math{\ifthenelse{\isempty{#2}}{\mathcal{O}_{#1}}}{{\mathcal{O}_{#1}^{#2}}}}
\newcommand{\Outer}[1]{\mathrm{O}^{\Pi}_{#1}}
\newcommand{\Inner}[1]{\mathrm{I}^{\Pi}_{#1}}
\newcommand{\Kernel}[1]{\mathrm{K}^{\Pi}_{#1}}

\newcommand{\qute}{\texttt{Qute}\xspace}
\newcommand{\qrup}{\texttt{qrp2rup}\xspace}
\newcommand{\depqbf}{\texttt{DepQBF}\xspace}
\newcommand{\qrattrim}{\texttt{QRAT-trim}\xspace}
\newcommand{\dqratchk}{\texttt{DQRAT-check}\xspace}

\tikzstyle{uredge}=[very thick,draw=red!80!green,line cap=round]
\tikzstyle{redge}=[line cap=round,very thick,draw=red!30!green!20!blue
]
\tikzstyle{axiomn}=[rectangle,very thick,draw=black!50,%
top color=white,bottom color=green!50,
minimum height=1.1em,minimum width=.5cm,inner sep=2pt,draw%
]

\tikzstyle{calcn}=[rectangle%
,rounded corners=2mm%
,thick%
,draw=black%
,top color=white,bottom color=black!20,%
minimum height=.5cm,minimum width=.5cm,inner sep=2pt%
]

\usetikzlibrary{calc}

\tikzstyle{infn}=[rectangle,rounded corners=1mm,thick,draw=black!50,%
top color=white,bottom color=black!15,%
minimum height=1.3em,minimum width=.5cm,inner sep=2pt,draw%
]

\tikzstyle{botn}=[rectangle,rounded corners=1mm,very thick,draw=black,%
top color=white,bottom color=red!25,%
minimum height=.5cm,minimum width=.5cm,inner sep=2pt,draw%
]

\tikzstyle{expcalcn}=[rectangle%
,thick%
,draw=green!40!black
,top color=white,bottom color=green!20,%
minimum height=.5cm,minimum width=.5cm,inner sep=2pt%
]

\tikzstyle{algcalcn}=[rectangle%
,thick%
,draw=yellow!40!black
,top color=white,bottom color=yellow!20,%
minimum height=.5cm,minimum width=.5cm,inner sep=2pt%
]

\tikzstyle{strongcalcn}=[rectangle%
,thick%
,draw=blue!40!black
,top color=white,bottom color=blue!20,%
minimum height=.5cm,minimum width=.5cm,inner sep=2pt%
]

\tikzstyle{cdclcalcn}=[rectangle%
,thick%
,draw=magenta!40!black
,top color=white,bottom color=magenta!20,%
minimum height=.5cm,minimum width=.5cm,inner sep=2pt%
]

\tikzstyle{expcalcn}=[rectangle%
,thick%
,draw=green!40!black
,top color=white,bottom color=green!20,%
minimum height=.5cm,minimum width=.5cm,inner sep=2pt%
]

\tikzstyle{strcalcn}=[rectangle%
,thick%
,draw=pink!40!black
,top color=white,bottom color=pink!20,%
minimum height=.5cm,minimum width=.5cm,inner sep=2pt%
]

\newlength{\radius}
\tikzstyle{decision} = [diamond, draw, fill=blue!20, 
text width=4.5em, text badly centered, node distance=3cm, inner sep=0pt]
\tikzstyle{block} = [rectangle, draw, fill=blue!20, 
text width=5em, text centered, rounded corners, minimum height=4em]
\tikzstyle{line} = [draw, -latex']
\tikzstyle{cloud} = [draw, ellipse,fill=red!20, node distance=3cm,
minimum height=2em]

\usepackage[utf8]{inputenc}
\usepackage{pgfplots}
\DeclareUnicodeCharacter{2212}{−}
\usepgfplotslibrary{groupplots,dateplot}
\usetikzlibrary{patterns,shapes.arrows}
\pgfplotsset{compat=newest}

\EventEditors{Alexey Ignatiev and Stefan Szeider}
\EventNoEds{2}
\EventLongTitle{29th International Conference on Theory and Applications of Satisfiability Testing (SAT 2026)}
\EventShortTitle{SAT 2026}
\EventAcronym{SAT}
\EventYear{2026}
\EventDate{July 20--23, 2026}
\EventLocation{Lisbon, Portugal}
\EventLogo{}
\SeriesVolume{377}
\ArticleNo{11}

\begin{document}
	\maketitle
	\begin{abstract}
		Certification for Quantified Boolean Formulas (QBF) and Dependency Quantified Boolean Formulas (DQBF) is an ongoing challenge. 
		Recent proof complexity work has shown that the majority of QBF and DQBF techniques can be p-simulated by using the independent extension rule. 

		In propositional logic, extension rules are supported by proof checkers using a more general RAT (Resolution Asymmetric Tautology) rule.
		The next step in (D)QBF certification would be to update these modern RAT formats to match the strength of this independent extension rule.

		In this paper we first introduce a new dependency scheme called \Dpu.
		This rule is the missing ingredient that when added to Blinkhorn's proof system DQRAT allows it to be provably p-equivalent to the Independent Extended QU-Res, the most powerful of the known QBF and DQBF proof systems. Up until now, DQRAT has only existed in theory, so we implement a prototype checker \texttt{DQRAT-check} which includes our extra rule. 
		
	In addition to its inclusion in our proof checker we show \Dpu has other properties that have been found for previous dependency schemes, and each of these observations has potential in solving/checking including the sound integration into the dependency learning solver \texttt{Qute}. 
	\end{abstract}

\section{Introduction}\label{sec:intro}
Our main motivation comes from the recent introduction of a new proof rule known as independent extension \cite{ChewPeitl25}, which has desirable certification properties. The central observation of this paper is that we can readily simulate independent extension, with only a quick ``fix'' of existing certification rules. The fix is a sound relaxation of the order of quantification and also works with solvers. 
In this paper we explore how, why and to what extent this works, before implementing it in a proof checker.

\subsection{Background}\label{sec:bgd}

The canonical $\NP$-complete problem is propositional satisfiability (SAT), and if we extend it with a prefix of Boolean quantifiers we get the Quantified Boolean Formula (QBF) problem which is then considered the canonical $\PSPACE$-complete problem. We can extend QBF even further into Dependency QBF (DQBF), where Boolean quantification still occurs, but the quantification order is not dictated by a linear prefix. Instead, every existentially quantified variable is given an explicit set of variables that it comes `after'. DQBF is considered a more expressive language than QBF as it is $\NEXP$-complete rather than in $\PSPACE$. 

While research into DQBF is studied for its own sake, we are very interested where DQBFs appear in QBF research, because making changes to the quantification order in a QBF can make a QBF easier to solve or prove. Occasionally such changes require us to shift to a non-linear quantifier prefix. One common example is when solvers and proof systems employ \emph{dependency schemes}, which calculate whether a dependency of one quantified variable to another is necessary or spurious. We can think of this as a transformation from a QBF to a DQBF as the prefix is recalculated to remove spurious dependencies. An example of well-studied dependency is the reflexive resolution path dependency scheme denoted by the symbol \Drrs. This detects that a dependency is spurious if there is no sufficient pair of resolution paths (see Section~\ref{sec:prelim} for the definition of a resolution path) from $\forall$ to $\exists$ literals.

The use of a whole range of transformations from QBF to DQBF, which is often done only implicitly, has made deciding a commonly agreed upon QBF proof format more challenging. Theoretical proof systems such as the sequent calculus \Gfull \cite{KP90}(named after Gentzen who pioneered sequent calculi) assume a neat symmetry between true and false as QBF is closed under negation and $\PSPACE$ is self-complementary, but DQBF is not alike in this way and has an inherent asymmetry between $\exists$ and $\forall$. In addition, theoretical proof systems make poor checking formats in a practical setting because they do not generalise existing practical formats and often proofs with polynomial size upper bounds are still too large in practice. 

One noteworthy QBF format that attempts to tackle this issue is \qrat (Quantified Resolution Asymmetric Tautology)\cite{HSB14}. In addition to generalising the rules of the propositional proof system \drat (Deletion Resolution Asymmetric Tautology), which is the closest system propositional logic has to a standard format, \qrat can detect the same side conditions that a dependency scheme detects. Instead of modifying the quantifier prefix, \qrat modifies a clause. Such a combination of rules ends up being incredibly powerful and \qrat can p-simulate a majority of the solving and preprocessing techniques in QBF, despite the large disparity of these techniques. This is incredibly fortunate overall, but it is not directly possible in \qrat to treat a QBF with a dependency scheme as a DQBF, as we soundly ought to be able to do. In fact, \qrat has been proven on a theoretical level to be no stronger than \Gfull \cite{CH22}, so even its detection of dependency schemes is no more than a sufficient combination of fundamental QBF steps. In an attempt to build \qrat into a genuine DQBF proof system, Blinkhorn proposed \cite{Blinkhorn2020SimulatingDP} the DQBF proof system \dqrat (Dependency Quantified Resolution Asymmetric Tautology) and 
\dqrat can handle dependency schemes directly. However, \dqrat has not been developed since its initial introduction, nor is there a substantial body of proof complexity follow up work that discusses it. 

Nonetheless, work into strong DQBF proof systems may yet be fruitful. 
Recent developments in QBF proof complexity have shown that a new DQBF proof system \dRes (where the power comes from an Independent Extension rule) can p-simulate practically everything in QBF and DQBF including \Gfull, \qrat, \dqrat and some dependency scheme rules such as \Drrs \cite{ChewPeitl25}. \dRes is not yet a practical format, and the proposed next step in the paper by Chew and Peitl introducing \dRes was to find a RAT version which was p-equivalent or more powerful.

In this paper we create a DQRAT proof checker and we add one additional rule to it. We demonstrate theoretically that adding this rule allows it to p-simulate  \dRes.
This rule is simply a new dependency scheme, which we name the pure universal dependency scheme (\Dpu). 

We cover the proof theoretical and proof complexity properties of \Dpu. Many of the properties, that we show, strengthen the case for \Dpu in various practical settings. For example, showing that \Dpu allows for strategy extraction in a resolution proof system shows that \Dpu can soundly be integrated into the solver \qute. 

As we have already emphasised, our major contribution is that \Dpu can strengthen \dqrat to be at least as strong as all other well-studied QBF and DQBF proof systems theoretical and practical. We avoid having to show each of these p-simulations individually, instead we show that \dqratpu and \dRes are p-equivalent, and rely transitively on the p-simulation results on \dRes shown in the paper by Chew and Peitl \cite{ChewPeitl25}.

\subsection{Related Work}\label{sec:rel}
QBF and DQBF proof complexity have been extensively studied, including work on QBF clause learning proof systems systems~\cite{BJ12,BWJ14,Bohm21}, expansion-based proof systems \cite{JM15,BCJ19,BCCM18} and stronger proof systems that go beyond current solving techniques~\cite{HSB14,BBCP20,KP90}.
One specific subtopic, the line of dependency scheme research, has progressed over the last decade.
The \emph{standard dependency scheme} was developed originally by Samer and Szeider~\cite{SamerS09}.
The \emph{reflexive resolution path dependency scheme} was developed by Slivovsky and Szeider \cite{Slivovsky-sat14}.
The \emph{tautology-free dependency scheme} was developed by Beyersdorff, Blinkhorn and Peitl \cite{BBP20}. 
Later the same set of authors developed a framework of an infinite number of  \emph{implication-free dependency schemes} \cite{BBP24}. The schemes progress in difficulty and trend towards conceptual dependency schemes that may not have polynomial-time checkability.

The main motivation of this work was to capture the power of independent extension clauses \cite{ChewPeitl25} for DQBF. Independent extension generalises weaker forms of QBF extension in QBF \cite{Jus07,BBCP20}. Independent extension involves conditioning on assignments, and similar ideas have been employed in other works such as conditional autarkies~\cite{KHB19,KullmannShukla19}.
There are many generalisations that capture extension clauses in other domains. 
Blocked clause elimination and addition considers redundant clauses that can be safely added or removed without changing satisfiability \cite{Kul99blocked}. Resolution asymmetric tautologies (RAT) \cite{HHW13} generalise blocked clauses, and propagation redundancies generalise RAT even further \cite{HKB17}.

This work involves the improvement of \dqrat to \dqratpu. \dqrat \cite{Blinkhorn2020SimulatingDP} is the result of generalising RAT addition rules in DQBF. 
Previously RAT was generalised into QBF through QRAT \cite{HSB14} and QRAT+ \cite{LE18}, these are QBF proof systems, but specifically designed for certification. As an alternative to these formats, the \qrp format~\cite{NPLSB12} can be used to provide more straightforward resolution type proofs, and recent work~\cite{PS25} aims to expand the potential of \qrp proofs.

\section{Preliminaries}\label{sec:prelim}

\subsection{Propositional Logic}
We assume the reader is familiar with propositional logic, literals and conjunctive normal form (CNF). We use the $\bar x$ notation to switch between a negated and non-negated literals and acts as an involution:
$\bar 0=1$ , $\bar 1 = 0$, $\bar x = \neg x$ and $\overline{(\neg x)}= x$. 

For CNF $\phi$, $\var(\phi)$ is the set of variables appearing in $\phi$. A partial assignment $\alpha$ on $\phi$ is a partial function that maps $\var(\phi)$ to $\{0,1\}$. We consider the restriction of $\phi$: $\phi|_\alpha$ to be the copy of $\phi$ where if $\alpha(x)=0$ literal $x$ is replaced by $0$ and literal $\neg x$ is replaced by $1$ and likewise if 
$\alpha(x)=1$ literal $x$ is replaced by $1$ and literal $\neg x$ is replaced by $0$.
We can overload $\alpha$'s functional notation to include literals so $\alpha(\neg l)= \neg \alpha(l) $.
Additionally we can represent a partial assignment as a string of literals, i.e. $x\bar z$ is the partial assignment that maps $x$ to $1$ and $z$ to $0$.

The negation of a clause $C$, denoted by $\bar C$ or $\neg C$ is
a partial assignment, but we can treat it syntactically as a CNF containing only the unit clauses of each of $C$'s literals but negated. 
Given the equivalence between partial assignments, negated clauses and sets of unit clauses,  
we can simplify a CNF $\phi$ that contains unit clauses, by applying the literal of a unit clause as an assignment, removing satisfied clauses and removing falsified literals. This process  can be repeated until no unit clauses remain and we defined this end point as \texttt{UP}$(\phi)$.

\begin{definition}
	We say $\phi\vdash_1 \bot$ if the resulting CNF of \texttt{UP}$(\phi)$ contains the empty clause.
\end{definition}

We can soundly derive clause $C$ into a CNF $\phi$ if $\phi\wedge \bar C$ is a contradiction. Checking for a contradiction is $\CoNP$-complete, so we use the weaker condition    $\phi\wedge \bar C \vdash_1 \bot$. This is often known as reverse unit propagation, here we call it ATA (asymmetric tautology addition).

\subsection{Quantified Boolean Formulas}
A quantified Boolean formula in closed prenex conjunctive form $\Pi \phi$, contains a CNF $\phi$ and quantifier prefix $\Pi$.
$\Pi$ is a sequence of pairs containing a quantifier symbol from $(\forall, \exists)$ and a propositional variable, i.e. $\Pi= \mathcal{Q}_1 x_1 \dots \mathcal{Q}_k x_k$ where $\mathcal{Q}_i\in \{\forall, \exists\}$ for $1 \leq i \leq k$. $\var(\Pi)$ is the set of variables appearing in prefix $\Pi$. 
Every variable in $\var(\phi)$ occurs exactly once in the prefix $\Pi$ (i.e. $\var(\phi)\subseteq \var(\Pi)$, we only consider closed sentences).

$\var_\exists(\Pi)$ is the set of existential variables (variables bound by $\exists$) appearing in prefix $\Pi$ and $\var_\forall(\Pi)$ is the set of universal variables (variables bound by $\forall$) appearing in prefix $\Pi$.
For assignment $\alpha$, $\Pi\restriction_{\alpha}$ removes all variables from $\dom(\alpha)$ and their attached quantification from the prefix. 
The prefix order $\lesssim_\Pi$ is a total pre-ordering of all variables in the prefix. $ x \lesssim_\Pi y$ if and only if $x$ appears left of $y$, or $x$ and $y$ are in the same uninterrupted contiguous block of variables with the same quantifier symbol. 

Because we are working with closed prenex QBFs, the semantics can  be defined using Skolem functions. We say that an existential variable $x$ depends on $u$ if $u$ is quantified left of $x$ in the prefix. We call $(u,x)$ a \emph{dependency pair}. In this way we can build a \emph{dependency set} $\D{x}$ for each existentially quantified variable $x$ containing exactly the universal variables that are left ($\lesssim_\Pi$) of $x$. A \emph{Skolem function} for an existential variable $x$ is a Boolean function $f_x:\{0,1\}^{\D{x}} \rightarrow \{0,1\}$. A QBF is true if and only if there is a set of Skolem functions for each existential variable, such that the Skolem functions together would satisfy the propositional matrix under all complete assignments to $\var_{\forall}(\Pi)$. 

You can imagine a QBF as a game between $\exists$ and $\forall$ in which they set the values of their variables from left to right in the prefix. $\exists$ wins if and only if the matrix evaluates to $1$ at the end of the game. The QBF is true if and only if $\exists$ has a winning strategy. Essentially this is no different to the idea of a set of satisfying Skolem functions.
This works dually for $\forall$ and falsity and we name the dual of Skolem functions as Herbrand functions. 

If we have a QBF $\Pi \phi$ and $\phi$ contains a clause $C\vee u$ where $\var(u)$ is universal and not contained in the dependency set of any of the variables of $C$ nor is there any $\bar u$ literal in $C$, then we can derive the clause $C$. This is because any winning $\exists$ strategy that satisfies $C\vee u$ will do so before arriving at $u$, hence $C$ must be satisfied by the same winning strategy.
\begin{prooftree}
	\AxiomC{$C \vee u$}
	\RightLabel{(UR)}
	\UnaryInfC{$C$ }
\end{prooftree}

One of the earliest observations about QBF was that if we combine the universal reduction rule with resolution, even when the resolution rule is restricted to cutting over existential literals only, then we get a complete refutational proof system for QBF known as \qrc \cite{KBKF95}.

\subsection{Dependency Quantified Boolean Formulas}
Dependency quantified Boolean formulas (DQBF) extend the language of QBF. We will concentrate on the S-form DQBFs (Skolem-form).
Like a QBF, every DQBF has two parts, a prefix and propositional matrix in CNF.
Also like a QBF, in a DQBF every variable is either existentially or universally quantified in said prefix.
In an S-form DQBF, each existential variable $x$ has an arbitrary dependency set $\D{x}$, the only restriction is that $\D{x}$ is a subset of the universal variables in the prefix $\Pi$.
The prefix $\Pi$ is written as $\Pi= \forall u_1 \dots u_p \exists {x_1}(D_{x_1}) \dots {x_q}(D_{x_q})$. Typically we omit the parentheses of the dependency set and list the elements.
In a DQBF prefix the written ordering does not determine anything, because the essential information is all contained in the specifications of the dependency sets. Nonetheless any QBF can be represented by a DQBF by specifying the QBF's dependency sets explicitly.  

A DQBF $\Phi$ is true if and only if there is a set $f = \{ f_x : x \in \var_\exists(\Phi) \}$ of Skolem functions, one for each $\exists$ variable $x$, such that for every complete assignment $\beta$ to all the universal variables the universal assignment completed with the values of the Skolem functions under that assignment, written as $\beta \cup f(\beta)$, form a satisfying assignment to the propositional matrix.
For example $$\forall u \forall v \exists x(u) \exists y(v) (v \vee x)\wedge (\bar v\vee \bar x)\wedge (u \vee y)\wedge (\bar u \vee \bar y)$$ is false because the $x$ Skolem function must satisfy the first clause when $v$ is false and the second when $v$ is true, but can only respond to $u$ and not $v$. But $$\forall u \forall v \exists x(v) \exists y(u) (v \vee x)\wedge (\bar v\vee \bar x)\wedge (u \vee y)\wedge (\bar u \vee \bar y)$$ is true because the $x$ Skolem function can be $\bar v$ and the $y$ Skolem function can be $\bar u$.
We call such a set of Skolem functions a \emph{model}.

We sometimes informally write $\forall U \exists E \phi$ for an arbitrary S-form DQBF, where $U$ is the set of universal variables, $E$ the set of existential variables each with their own dependency set that we may hide for the time being, and $\phi$ a propositional matrix containing no quantifiers.
We can define a subprefix and the relation $\Pi \subseteq\Omega$, whenever $\var_\exists(\Pi)\subseteq\var_\exists(\Omega)$, $\var_\forall(\Pi)\subseteq\var_\forall(\Omega)$ and if  $x\in \var_\exists(\Pi)$ then $\D{x} = \mathrm{D}_{x}^\Omega$.

\subsubsection{Pre-ordering a DQBF}\label{sec:outer}
Recall that a QBF prefix $\Pi$ has a total pre-ordering ($\lesssim_\Pi$) whose equivalence classes are quantifier blocks.
Instead let $\Pi$ be a DQBF prefix. It turns out we can still define a pre-ordering on  $\Pi$, but it may not be total.  This allows QRAT to be generalised to the original definition DQRAT, which is the base system we aim to improve. For readers purely interested in the new dependency rule we introduce later, understanding of this pre-ordering is not needed, but it is important for the wider topic of DQBF certification. 

First we overload the $\D{}$ notation. For universal variables $u$ we say $\D{u}=\{u\}$.
For formulas $\phi$  (including partial assignments) we consider the dependency set $\D{\phi}$ to be $\bigcup_{x\in \var(\phi)} \D{x}$. 

We define the set of outer variables for each variable in $\Pi$. For an existential variable $x$ in $\Pi$, the outer variables
are the set of variables that can be calculated from the Skolem functions using only the values for the dependency set. The outer variable set $\Outer{x}:= \{y\in \var(\Pi) \mid \D{y}\subseteq\D{x} \}$.
For a universal variable $u$ to find its outer variables, we look at its inner existential variables, those existential variables that contain $u$ in its dependency set i.e. $\Inner{u}:=\{y\in \var(\Pi) \mid \D{u}\subseteq\D{y}\}$, and we include any universal variable that is common in all of these, we also include any existential variable that depends only on these common universal variables, though we make an exception and exclude it, if it depends on $u$, in order to keep the notion the $\forall$ variable comes before the $\exists$ variables that depend on it. $\Outer{u}= \{u\}\cup  \bigcap_{x\in \Inner{u}} \{y \mid \D{y} \subseteq \D{x}\setminus\{u\}\}$.
With a DQBF prefix $\Pi$ we can define the relation $\lesssim_\Pi$, so $x \lesssim_\Pi y$ means $x\in \Outer{y}$, or equivalently that $\Outer{x}\subseteq \Outer{y}$.

\begin{example}
	Consider the prefix:
	$\forall u, v, w \exists a(u,v,w), b(), c(u,v), d(u,w), e(u)$.
	We gain the following sets of outer variables (Figure~\ref{fig:outer}):
	\begin{figure}[h]
		\centering
		\begin{tikzpicture}
			\node(a) at (0,5){$\Outer{a}=\{u,v,w,a,b,c,d,e\}$};
			\node(c) at (-2,4){$\Outer{c}=\{u,v,b,c,e\}$};
			\node(d) at (2,4){$\Outer{d}=\{u,w,b,d,e\}$};
			\node(v) at (-2,3){$\Outer{v}=\{u,v,b,e\}$};
			\node(w) at (2,3){$\Outer{w}=\{u,w,b,e\}$};
			\node(e) at (0,2){$\Outer{e}=\{u,b,e\}$};
			\node(u) at (0,1){$\Outer{u}=\{u,b\}$};
			\node(b) at (0,0){$\Outer{b}=\{b\}$};
			\draw(b)--(u)--(e);
			\draw(e)--(v)--(c)--(a);
			\draw(e)--(w)--(d)--(a);
		\end{tikzpicture}
		\caption{Hasse diagram of equivalence classes of outer variables of an example DQBF prefix.\label{fig:outer}}
	\end{figure}
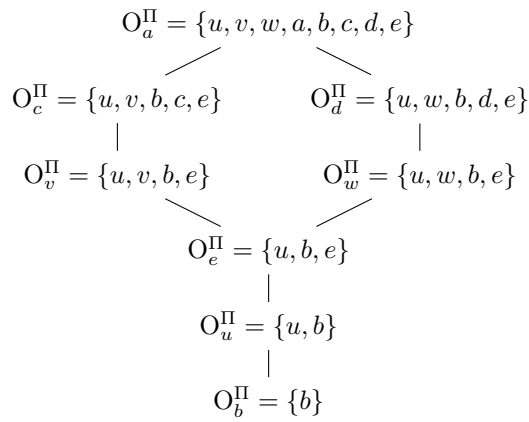
\end{example}


\subsubsection{Dependency Schemes}\label{sec:depscheme}

Dependency schemes are a method of evaluating individual dependency pairs to see if they are really necessary.  
Let $\Pi \phi$ be a DQBF (or QBF). 
We consider pairs $(u,x)$, with $u$ a universal variable and $x$ an existential variable.  \emph{Spurious} dependencies occur when $u\in \D{x}$ but it does not change the truth of the DQBF to remove $u$ from  $\D{x}$. Given a DQBF or QBF $\Pi\phi$ the trivial dependency scheme $\Dtrv(\Pi\phi)$ is the set of all pairs $(u,x)$ where $u$ is universal, $x$ is existential and $u\in \D{x}$.

The most successful dependency schemes analyse resolution paths between clauses.
We can represent a CNF as a multi graph where nodes are clauses. There is an edge from $C_1$ to $C_2$ whenever there is a variable $p$  such that $p$ is in one clause and $\neg p$ is in the other. We annotate the edge with variable $p$. 
A \emph{resolution path} is a sequence of edges $\{e_i \mid 1\leq i\leq k\}$ that form a path such that no two consecutive edges are annotated with the same variable. We can list a path using the clauses, and the literals in-between.
If there is no resolution path between two clauses $D$ and $E$ for a CNF $\phi$, then for any resolution refutation of $\phi$ both $D$ and $E$ cannot occur in the connected part of the proof. 
If we want to analyse variables, particularly ordered variables like in QBF, it is useful to restrict a resolution path to a particular subset of variables $\mathcal{S}$. So that  $\{e_i \mid 1\leq i\leq k\}$ is a resolution path if all edges are annotated with a variable from $\mathcal{S}$ and all consecutive edges coincide on a clause but have different annotations.

One of the most well used non-trivial dependency schemes is the \emph{reflexive resolution path dependency scheme}. To compute which dependencies actually are required, which we denote as $(u,x)\in \Drrs(\Pi \phi)$, we use the notion of a resolution path. 
Let $\psi$ be a DQBF, $\chi$ a subset of the clauses in $\psi$ and $\mathcal{S}$ a subset of variables appearing in $\psi$. We define using a fixpoint a set of clauses $\pathc{rrs}(\psi, \chi, \mathcal{S})$ as well as a set of literals $\pathl{rrs} (\psi, \chi, \mathcal{S})$.
Initially, $\pathc{rrs}(\psi, \chi, \mathcal{S})$ contains all clauses from $\chi$ and $\pathl{rrs} (\psi, \chi, \mathcal{S})$ contains all $\mathcal{S}$-literals in $\chi$.
We expand these sets in the following way:
suppose $p\in \pathl{rrs} (\psi, \chi, \mathcal{S})$, we include any $\psi$ clause $E$ with $\bar p\in E$ to $\pathc{rrs}(\psi, \chi, \mathcal{S})$ and include the literals $\{x\in E \mid x\neq \bar p, x \in \mathcal{S}\}$ to $\pathl{rrs} (\psi, \chi, \mathcal{S})$. $\bar p$'s non-inclusion is why we cannot just list the clauses.

We define $\phi_u$ to be the set of clauses that contain $u$, and $\phi_{\bar u}$ to be those that contain $\bar u$, (assume no clauses are tautological). Let $\mathcal{S}_u$ be the set of existential variables that contain $u$ in its dependency set. 
We say $u$ has a reflexive resolution path to literal $x$ (denoted $u\rrssim x$) if $x\in \pathl{rrs} (\Pi \phi, \phi_u, \mathcal{S}_u)$, likewise $\bar u\rrssim x$ if $x\in \pathl{rrs} (\Pi \phi, \phi_{\bar u}, \mathcal{S}_u)$.
The purpose of $u\rrssim x$ is to identify that the existential player may be required  to satisfy an $x$ literal while dealing with a falsified $u$ literal further back along the reflexive resolution path.

For the reflexive resolution path dependency scheme,
 $(u,x)\in \Drrs(\Pi \phi)$  if $u\in \D{x}$ and either:
 $u \rrssim x$ and $\bar u \rrssim \bar x$, or $\bar u \rrssim x$ and $u \rrssim \bar x$.

We sometimes omit $\Pi \phi$ when the DQBF is clear i.e. $(u,x)\in \Drrs$. The reader should interpret membership  $(u,x)\in \Drrs$ as dependence, and $(u,x)\notin \Drrs$ as independence.

\subsubsection{DQBF Proof Systems}
While \qrc has been proven incomplete for DQBFs it is still sound \cite{balabanov2014henkin}, so we can use parts of \qrc in a refutation. 
In order to make it complete, so that every false DQBF has a refutation leading to the empty clause,  Chew and Peitl \cite{ChewPeitl25} discovered an extension rule that created new variables with minimal dependency sets. 
This powerful calculus is known as \dRes and is given in Figure~\ref{fig:dres}.

\begin{figure}[ht]
	\framebox{\parbox{0.95\textwidth}{

			\begin{prooftree}
				\AxiomC{}
				\RightLabel{(Ax)}
				\UnaryInfC{$L$}
				\DisplayProof\hspace{1cm}
				\AxiomC{$C \vee u$}
				\RightLabel{(Red)}
				\UnaryInfC{$C$}
				\DisplayProof\hspace{1cm}
				\AxiomC{$E \vee \neg x$}
				\AxiomC{$F \vee x$}
				\RightLabel{(Res)}
				\BinaryInfC{$E \vee F$}
			\end{prooftree}
			
			$L$ is a clause in the propositional matrix $\phi$. 
			$u$ is a $\forall$ literal. There is no $\exists$ literal $l$ in $C$ such that $\var(u)\in \D{\var(l)}$, and there is no $\bar u\in C$.
			
			\begin{prooftree}
				\AxiomC{}
				\RightLabel{(IndExt)}
				\UnaryInfC{$(\bar\alpha \vee v \vee  y_1),(\bar\alpha \vee v \vee   y_2),(\bar\alpha \vee  \bar v \vee \bar y_1 \vee \bar y_2)  $}
			\end{prooftree}
			$v$ is a fresh $\exists$ variable, $\alpha$ is a conjunction of $\forall$ literals. $\D{v}= (\D{y_1} \cup \D{y_2} ) \setminus \D{\alpha}$.
			\medskip\\
			As an additional rule, 
			the prefix $\Pi$ may be weakened to $\Pi'$ to add a new variable.
	}}
	\caption{Proof rules of \dRes \label{fig:dres}.}
\end{figure}

\begin{example}\label{ex:dres proof}
		Consider the following DQBF from~\cite[Theorem~7]{balabanov2014henkin}:
	\begin{align*}
		\forall u \forall v  \; \exists x(u) \exists y(v)\;
		(u \vee x \vee y)\; \wedge \; (\bar u  \vee \bar v\vee  x \vee y) \; \wedge \;
		(\bar u  \vee  v\vee  x \vee \bar y)\\ 
		\wedge \; (u \vee \bar x \vee \bar y)\; \wedge \; (\bar u  \vee \bar v\vee  \bar x \vee \bar y) \; \wedge \;
		(\bar u  \vee  v\vee  \bar x \vee y) 
	\end{align*}
	
	We can refute it by first adding the three clauses:
	$(\bar u \vee n \vee  x),(\bar u \vee n \vee   x),( \bar u \vee  \bar n \vee \bar x \vee \bar x)  $ using the IndExt rule, although since we are only using one argument $x$, two of these clauses are identical and one can be simplified with idempotence.
	We can interpret the clauses as conditionally defining $n$,  ($u\rightarrow (n=\bar x)$), because $n$ is defined only when $u$ is true, it can without penalty assume any value when $u$ is false. This allows \dRes to soundly let $n$ not depend on $u$ despite $x$ doing so, therefore $\D{n}=\{\}$.
	
	We can resolve $(\bar u \vee n \vee   x)$ with two clauses to get $(\bar u  \vee \bar v\vee  n \vee \bar y)$ and $
	(\bar u  \vee  v\vee  n \vee y) $ which reduce to $( \bar v\vee  n \vee \bar y)$ and $
	( v\vee  n \vee y) $. 
	Likewise, we can resolve our other extension clause $( \bar u \vee  \bar n \vee  \bar x)$ with two axiom clauses to get $( \bar u \vee \bar v\vee  \bar n \vee y)$ and $(\bar u \vee v\vee  \bar n \vee \bar y)$ which become $( \bar v\vee  \bar n \vee y)$ and $( v\vee  \bar n \vee \bar y)$ after reduction.
	We show the remaining derivation as a DAG to the empty clause in Figure~\ref{fig:ex_dres}. Note that the rules of \qrc have been shown to be inadequate to derive a contradiction in these formulas\cite{balabanov2014henkin}. 

\begin{figure}
	\centering{
	\begin{tikzpicture}[xscale=1.0,yscale=0.8]

		\node[infn](r1) at (-3,3){$  v\vee  \bar n \vee \bar y$};
		\node[infn](r2) at (-5, 3){$ \bar v\vee  \bar n \vee  y$};
		\node[infn](r3) at (3, 3){$ v\vee  n \vee y$};
		\node[infn](r4) at (5, 3){$ \bar    v\vee  n \vee \bar y$};
		\node[axiomn](ax5) at (1,3){$u \vee x \vee y$};
		\node[axiomn](ax6) at (-1,3){$u \vee \bar x \vee \bar y$};

		\node[infn](e1) at (-1.5,2){$ u\vee v\vee  \bar n \vee x$};
		\node[infn](e2) at (-4, 2){$ u\vee \bar v\vee  \bar n \vee  \bar x$};
		\node[infn](e3) at (1.5, 2){$ u\vee  v\vee  n \vee \bar x$};
		\node[infn](e4) at (4, 2){$ u\vee    \bar v\vee  n \vee x$};
		
		\draw[red](ax5)--(e1)--(r1);
		\draw[red](ax6)--(e2)--(r2);
		\draw[red](ax6)--(e3)--(r3);
		\draw[red](ax5)--(e4)--(r4);
		
		\node[infn](f1) at (-1.5,1){$ u\vee  \bar n \vee x$};
		\node[infn](f2) at (-4, 1){$ u\vee  \bar n \vee  \bar x$};
		\node[infn](f3) at (1.5, 1){$ u\vee   n \vee \bar x$};
		\node[infn](f4) at (4, 1){$ u\vee     n \vee x$};
		
		\draw[blue](e1)--(f1);
		\draw[blue](e2)--(f2);
		\draw[blue](e3)--(f3);
		\draw[blue](e4)--(f4);
		\node [infn](g1) at (-2.5,0){$ u\vee     \bar n$};
		\node [infn](g2) at (2.5,0){$ u\vee     n$};
		\draw[red](f1)--(g1)--(f2);
		\draw[red](f3)--(g2)--(f4);
		\node[infn](h1) at (-2.5,-1){$\bar n$};
		\node[infn](h2) at (2.5,-1){$n$};
		\draw[blue](g1)--(h1);
		\draw[blue](g2)--(h2);
		\node[infn](end) at (0,-2){$\bot$};
		\draw[red](h1)--(end)--(h2);
	\end{tikzpicture}
	\caption{Proof DAG of the final \textcolor{red}{resolution} and \textcolor{blue}{reduction} steps in Example~\ref{ex:dres proof}\label{fig:ex_dres} and Example~\ref{ex:dqrat}.}
}
\end{figure}
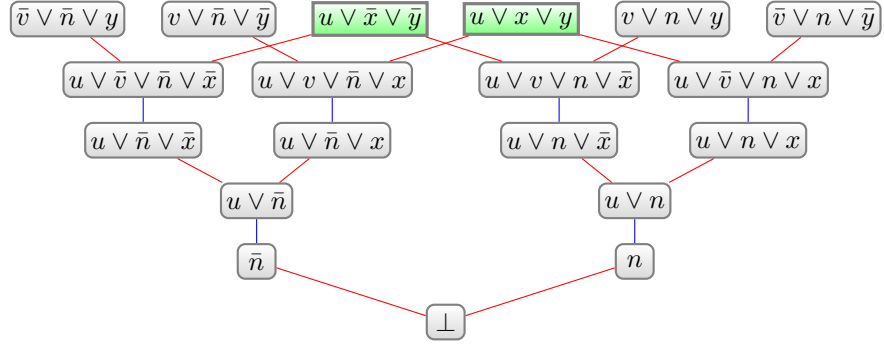
\end{example}

In reality, finding the optimal uses of the IndExt rule requires a high degree of non-determinism. 
\dRes is a DQBF generalisation of Extended Resolution in propositional logic, which similarly allows the addition of extension clauses as a rule. In propositional logic most checkers have progressed from an extension clause checker to a RAT addition checker, so we are interested in the one RAT based DQBF proof system: \dqrat. 

\dqrat is a combination of rules that have complicated checking criteria (Figure~\ref{fig:dqrat}). The attempt here is to generalise what was practically checked in propositional proofs, but we have to bring in both the DQBF notion of outer clauses 
from Section~\ref{sec:outer} 
and the reflexive resolution path from Section~\ref{sec:depscheme}.

\begin{figure}[h]
	\framebox{\parbox{0.95\textwidth}{
			In all rules, let $\Pi, \Omega$ be DQBF prefixes, $\phi$ be a CNF, $C$ be a clause and $l$ be a literal with $\var(\phi), \var(C), \var(l)$ assumed to be subsets of $\var(\Pi)$.
			
			\begin{prooftree}
				\AxiomC{$\Pi \phi$}
				\RightLabel{(ATA)}
				\UnaryInfC{$\Pi \phi \wedge  C$}
				\DisplayProof\hspace{1cm}
				\AxiomC{$\Pi \phi\wedge  C$}
				\RightLabel{(Del)}
				\UnaryInfC{$\Pi \phi $}
				\DisplayProof\hspace{1cm}
				\AxiomC{$\Pi \phi\wedge  (C\vee l)$}
				\RightLabel{(UR)}
				\UnaryInfC{$\Pi \phi \wedge  C$}
			\end{prooftree}
			\textbf{ATA}: $\phi\wedge \bar C \vdash_1 \bot$ is required. \qquad \qquad\quad\textbf{Del}: there is no side condition on $C$. \\ \textbf{UR}: we require that $l$ is universal and $\var(l) \notin \D{C}$. 
			
			\begin{prooftree}
				\AxiomC{$\Pi \phi$}
				\RightLabel{(DQRAT$_\exists$)}
				\UnaryInfC{$\Pi \phi \wedge  (C\vee l)$}
				\DisplayProof\hspace{1cm}
				\AxiomC{$\Pi \phi\wedge  (C\vee l)$}
				\RightLabel{(DQRAT$_\forall$)}
				\UnaryInfC{$\Pi \phi \wedge  C$}
			\end{prooftree}
			
			\textbf{DQRAT}: 
			$l$ is existential in DQRAT$_\exists$, universal in DQRAT$_\forall$.
			For all clauses $D$ in $\phi$ with $\bar l\in D$, the following must hold: \\
			$\bm{(\exists)}$: $\phi\wedge \neg C \wedge \bar l \wedge\bigwedge_{x \in D, x \neq \bar l}^{\var(x)\lesssim_\Pi \var(l)} \bar x \vdash_1 \bot$.\;\quad 
			$\bm{(\forall)}$: $\phi\wedge \neg C \wedge l \wedge\bigwedge_{x \in D, x \neq \bar l}^{\var(x)\lesssim_\Pi \var(l)} \bar x \vdash_1 \bot$.
			\begin{prooftree}
				\AxiomC{$\Pi \phi$}
				\RightLabel{(BPM)}
				\UnaryInfC{$\Omega \phi$}
				\DisplayProof\hspace{1cm}
				\AxiomC{$\Pi \phi$}
				\RightLabel{(\Drrs)}
				\UnaryInfC{$\Omega \phi$}
			\end{prooftree}
			
		\textbf{BPM}: $\Pi \subseteq \Omega$. \\
		$\bm{\mathcal{D}}^{\textbf{rrs}}$: 	$\var_\exists(\Pi)=\var_\exists(\Omega)$, $\var_\forall(\Pi)=\var_\forall(\Omega)$, $u\notin \mathrm{D}_{x}^\Omega$ only if $(u,x) \notin \Drrs(\Pi\phi)$.
	}}
\caption{Proof rules of \dqrat\cite{Blinkhorn2020SimulatingDP}.\label{fig:dqrat}}
\end{figure}

\begin{example}\label{ex:dqrat}
	Once again we should look at the DQBF from Example~\ref{ex:dres proof}. We will first use BPM (Basic Prefix Modification) to create a new $\exists$ variable $n$ such that $\D{n}=\{u\}$. In terms of outer variables, $n$ is in the same equivalence class as $x$ and has outer variables $\{u,x,n\}$. 
	We can trivially add clause $ (n \vee x)$ via DQRAT$_\exists$ as no clause contains $\bar n$. We can then add $ (\bar n \vee \bar x)$ via DQRAT$_\exists$. This time it has to check against one clause with unit propagation but an easy contradiction between $x$ and $\bar x $ exists on the left hand side so this is immediate.
	 
	Now we add the 4 clauses $(\bar u  \vee \bar v\vee  n \vee \bar y)$, $
	(\bar u  \vee  v\vee  n \vee y)$, $( \bar u \vee \bar v\vee  \bar n \vee y)$ and $(\bar u \vee v\vee  \bar n \vee \bar y)$ which all can be done via ATA. 
	The next part is important because it diverges from the proof in \dRes. We need to delete the clauses $ (n \vee x)$,  $(\bar n \vee \bar x)$. Only then can we apply the \Drrs rule. 
	Since $u\nrrssim n$ and $u \nrrssim \bar n$, we can change the dependency set of $n$ from $\{u\}$ to the empty set. After this point we can reduce the 4 clauses we added via ATA to get $(\bar v\vee  n \vee \bar y)$, $
	(  v\vee  n \vee y)$, $(  \bar v\vee  \bar n \vee y)$ and $(v\vee  \bar n \vee \bar y)$. Since ATA can be used to add instances of the resolution rules we can proceed to follow the same proof as in the end of the \dRes proof (Figure~\ref{fig:ex_dres}). 
	A machine readable proof of this formula is available in Example~\ref{ex:dqratcheck} and also at  
\url{https://ac.tuwien.ac.at/files/peitl/sat2026/}.

\end{example}

\dRes and \dqrat are DQBF proof systems and can be compared via p-simulation. We say a proof system $f$ \emph{p-simulates} a proof system $g$ if there is a polynomial time procedure that maps $g$ proofs to $f$ proofs of the same theorem. A p-simulation is impossible if there is a separating family of theorems, whose minimum proof size in $f$ is bounded below by a super-polynomial function in the minimum proof size in $g$.
\dRes p-simulates \dqrat, but the converse is an open problem.

\section{A New Dependency Scheme}
In this section, we define an improvement on \Drrs, which we  call the \emph{pure universal dependency scheme}. 
The idea is that we can exclude a resolution path from $u$ to $x$ if it is merely an extension of a resolution path from $\bar u $ to $x$, as only the minimal path should be considered.

\subsection{Definition of the Pure Universal Dependency Scheme}

For a DQBF $\Pi\phi$, we will denote membership $(u,x)\in \Dpu(\Pi\phi)$ to mean that $\exists$-variable $x$ really does depend on $\forall$ variable $u$ after considering the pure universal dependency scheme.
\begin{definition}
Let $\psi$ be a DQBF, $\chi$ a subset of the clauses in $\psi$, $\mathcal{S}$ a subset of variables appearing in $\psi$ and $u$ a universal literal.
We define the sets $\pathcpu(u,\psi, \chi, \mathcal{S})$ and  $\pathlpu(u,\psi, \chi, \mathcal{S})$ as follows.
Initially $\pathcpu(u,\psi, \chi, \mathcal{S})$ contains all clauses from $\chi$ and $\pathlpu (u, \psi, \chi, \mathcal{S})$ contains all $\mathcal{S}$-literals in $\chi$.
We expand these sets in a similar way as before, suppose $p\in \pathlpu (u,\psi, \chi, \mathcal{S})$ we include any $\psi$ clause $E$ with $\bar p\in E$ \textbf{and crucially} $\bar u \notin E$ into $\pathcpu(u,\psi, \chi, \mathcal{S})$, and also include the literals $\{x\in E \mid x\neq \bar p, x \in \mathcal{S}\}$ into
$\pathlpu (u,\psi, \chi, \mathcal{S})$, and as we increase this set we can further propagate until we reach fix-point.
\end{definition}

\begin{definition}
	\label{def:upure-connection}
Given a DQBF $\Pi\phi$ with CNF matrix $\phi$.
We say there is a pure path from universal literal $u$ to existential literal $x$ (denoted $u\pusim x$) if and only if $x\in \pathlpu (u,\psi, \phi_u, \mathcal{S}_u)$.
Where $\phi_u$ is the set of clauses that contain literal $u$ and  $\mathcal{S}_u$ is the set of $\exists$ variables that contain $u$ in its dependency set. 
\end{definition}

\begin{definition}
	\label{def:upure-non-dependency}
	Given a DQBF $\Pi\phi$ with CNF matrix $\phi$, let $(u,x)$ be a pair with a $\forall$ variable $u$ and $\exists$ variable $x$. 
	$(u,x)\in \Dpu(\Pi \phi)$ ($\in\Dpu$ when unambiguous) if and only if $u\in \D{x}$ and either:
	\textbf{1.} $u \pusim x$ and $\bar u \pusim \bar x$, or \textbf{2.} $\bar u \pusim x$ and $u \pusim \bar x$.

\end{definition}

\begin{example}\label{ex:dpu}
	Consider the following QBF:
	\begin{align*}
		\forall u \forall v  \exists x \exists y \exists z\;
		(u \vee v \vee y)\; \wedge \; (u  \vee \bar v\vee  x \vee \bar y) \; \wedge \;
		(\bar u \vee v \vee z)\; \wedge \; (\bar u  \vee \bar v\vee  \bar x \vee \bar z).
	\end{align*}
	We can first make some observations that are true for both \Drrs and for \Dpu. $(v,x)\notin \Dpu$  and $(v,x)\notin \Drrs$  as $v\nrrssim x$ nor $v\nrrssim \bar x$,
	which means $v\npusim x$ nor $v\npusim \bar x$, it is always the case that $(v,x)\notin \Drrs$ entails $(v,x)\notin \Dpu$.
	However	$(v,y), (v,z)\in \Dpu$ as we can make immediate paths, likewise for $(u,x)\in \Dpu$.
	
	Where \Dpu differs from \Drrs can be seen for $(u,y)$ and $(u,z)$. $(u,y)$ is in \Drrs because $u\rrssim y$ and $\bar u \rrssim \bar y$ (through $\bar x$, first). However $\bar u \npusim \bar y$ as the path from $\bar u$ cannot use $(u  \vee \bar v\vee  x \vee \bar y)$ as it contains a positive $u$. We observe a similar situation for $(u,z)$ where $u\rrssim \bar z$ but as the path must go through $(\bar u  \vee \bar v\vee  \bar x \vee \bar z)$, $u\npusim \bar z$.

\end{example}

\subsection{The \texorpdfstring{\Dpu}{Dpure} Prefix Modification Rule}

Given a universal literal $u$, we can calculate the set of literals $ \pathlpu (u,\psi, \chi, \mathcal{S}_u)$ in linear time in the total number of individual literals appearing in $\psi$.
Once we have found the variables that lack sufficient paths for $\Dpu$, we can subtract $u$ from their dependency sets, and this will not affect the process if we repeat it for any different universal literals. 
Therefore it takes polynomial time in $\psi$ to completely recalculate the prefix according to $\Dpu$ and can be considered a polynomial-time checkable proof system. We prove it sound in this section.

\begin{example}
	We can take the QBF from Example~\ref{ex:dpu} and modify its prefix according to \Dpu to get the following DQBF:
	\begin{align*}
		\forall u \forall v \;  \exists x(u) \exists y(v) \exists z(v)\;
		(u \vee v \vee y)\; \wedge \; (u  \vee \bar v\vee  x \vee \bar y) \; \wedge \;
		(\bar u \vee v \vee z)\; \wedge \; (\bar u  \vee \bar v\vee  \bar x \vee \bar z). 
	\end{align*}
\end{example}

For the proof of Theorem~\ref{thm:dpu-soundness} recall that for a total universal assignment $\beta$ and a set of Skolem functions $f$ we define the completed assignment $\beta \cup f(\beta)$ by
\[ \bigl[ \beta \cup f(\beta) \bigr] (x) = \begin{cases}
	\beta(x) & x \in \var_\forall, \\
	f_x(\beta|_{\D{x}}) & x \in \var_\exists,
\end{cases}
\]
in other words, the assignment that is $\beta$ on universal variables and computes a value according to the respective Skolem function for existential variables.
\newcommand{\mpool}{\mathcal{M}}
\begin{theorem}[Soundness]
	\label{thm:dpu-soundness}
	Let $\Pi\phi$ be a true DQBF.
	Let $\Pi'$ be the prefix where 
	$\var_\exists(\Pi)=\var_\exists({\Pi'})$, $\var_\forall(\Pi)=\var_\forall({\Pi'})$, $u\in \mathrm{D}_{x}^{\Pi'}$ if and only if $(u,x) \in \Dpu(\Pi\phi)$. Then 
	 $\Pi'\phi$ is true.
\end{theorem}

\begin{proof}
	Let $f=\{ f_e \mid e\in \var_\exists(\Pi)\}$ be a set of Skolem functions that satisfies $\Pi \phi$. 
	We will construct a set of Skolem functions $f^*=\{ f^*_e \mid e\in \var_\exists(\Pi)\}$ that satisfies $\Pi' \phi$.
	
	We use the following definitions. If $\alpha$ is partial assignment to the universal variables, $\alpha^v$ is obtained by flipping the Boolean value of universal variable $v$. Given $f$ a \emph{dependency witness} is a  triple $(\alpha, y, v)$ where $f_y(\alpha)\neq  f_y(\alpha^v)$, $y$ is an existential variable, $v$ is a universal variable, $\alpha$ is an assignment to $\D{y}$. The set $\mathcal{S}_u=\{z:u \in \D{z}\}$. For universal literal $l$: $\phi_l= \{C \in \phi \mid l\in C\}$.
	For an existential variable $z$, and assignment $\gamma\in \{0,1\}^U$: $\gamma_z= \gamma|_{\D{z}}$.

	Suppose $f$ has a dependency witness $(\alpha, x, u)$ where $u \in \mathrm{D}_x^\Pi \setminus \mathrm{D}_x^{\Pi'}$.
	Let $l_u$ be the literal on $u$ satisfied by $\alpha$, $l_x$ the literal on $x$ satisfied by $f_x(\alpha)$.
	Since $u \not \in \mathrm{D}_x^{\Pi'}$, by definition either $\bar l_u \not\pusim l_x$ or $l_u \not\pusim \bar l_x$.
	Without loss of generality, let $\neg l_u \not\pusim l_x$ (in the other case, swap the roles of $\alpha$ and $\alpha^u$).
	
	Define $\mpool = \mpool(f, \bar l_u, l_x, \alpha)$ as the set of all candidate Skolem sets $f'=\{f'_x \mid  x\in \var_\exists(\Pi)\}$ respecting the prefix $\Pi$ with the following properties:
	\begin{enumerate}
		\item The set of dependency witnesses in $f'$ and $u$ are a \emph{proper} subset of dependency witnesses in $f$ and $u$. 
		\label{item:hairiness}
		\item For every existential variable $z$, for every
		$ \gamma_z \in \{0,1\}^{\D{z}}$, if  
		$f'_z (\gamma_z) \neq f_z(\gamma_z)$  then $\gamma_z$ satisfies $l_u$; i.e. $u\in \D{z}$ and $l_u$'s polarity matches with $\gamma_z$.
		\label{item:assignment}
		\item For every existential variable $z$, for every $\gamma_z \in \{0,1\}^{\D{z}}$, if  
		$f'_z (\gamma_z) \neq f_z(\gamma_z)$  then
		for the literal $l_z$ satisfied by $f_z(\gamma_z)$,
		$\lnot l_u \not \pusim l_z$.
		\label{item:path}
	\end{enumerate}
	We first show that $\mpool$ is non-empty. Let
	\[ f^0=\{f^0_z \mid z \in \var_\exists(\Pi) \}\quad
	f^0_z(\tau) = \begin{cases}
		\neg f_x(\tau) & \text{ if } \tau = \alpha \text{ and } z=x\\
		f_z (\tau)     & \text{ otherwise,}
	\end{cases}
	\]
	\begin{enumerate}
		\item
		$f^0$ satisfies property~
		\ref{item:hairiness} because $(\alpha, x, u)$ and $(\alpha^u, x, u)$ are dependency witnesses that are removed. 
		For any dependency witness $(\delta, w, u)$ present in $f^0$ then $w\neq x$ or ($\delta\neq \alpha$ and $\delta\neq \alpha^u$ ) and so $f^0_w(\delta)=f_w(\delta)$ and $f^0_w(\delta^u)=f_w(\delta^u)$ so $(\delta, w, u)$ is a witness in $f$.
		\item $f^0$ only differs from $f$ on $\alpha$ and its extensions, but $\alpha$ satisfies $l_u$.
		\item $f_z^0$ only differs from $f_z$ when $z=x$ however we have assumed that $\lnot l_u \not \pusim l_x$ and $f_x(\alpha)$ satisfies $l_x$. 
	\end{enumerate}
	Next we  create an inductive step. We show that if $f'\in \mpool$ is not a model there is $f''\in \mpool$ where the dependency witnesses of $f''$ and $u$ are a proper subset of the  witnesses of $f'$.
	
	First we need to construct such a $f''$. We know that if $f'$ is not a model of $\Pi \phi$, then there is some assignment $\gamma\in \{0,1\}^{\var_\forall(\Pi)}$ and some clause $C\in \phi$ such that the assignment $\gamma \cup f'(\gamma)$ falsifies $C$. 
	However we know that $\gamma \cup f(\gamma)$ satisfies $C$ because $f$ is a model. Hence there is some existential variable $z$ such that $f'_z(\gamma_z)\neq f_z(\gamma_z)$ and some literal $l_z$ in $z$ such that $f_z(\gamma_z)$ satisfies $l_z$ and $f'_z(\gamma_z)$ falsifies $l_z$. By property~\ref{item:assignment}: $ \gamma_z$ satisfies $l_u$, and by property~\ref{item:path}: $\bar l_u \not \pusim l_z$.
	On $\gamma^u$, we have that $l_u$ is falsified so $f'$ and $f$ are equal by property~\ref{item:assignment}.  This means that $f'_z(\gamma^u_z)=f_z(\gamma^u_z)$. Therefore if $f_z(\gamma_z)\neq f_z(\gamma^u_z)$ (i.e. $(\gamma_z, z, u)$ is not a dependency witness in $f$) then 
	$f'_z(\gamma_z)\neq f_z(\gamma_z)=f_z(\gamma^u_z)=f'_z(\gamma^u_z)$ (i.e. $(\gamma_z, z, u)$ is a dependency witness in $f'$) violating property~\ref{item:hairiness}, so $f_z(\gamma^u_z)\neq f_z(\gamma^u_z)$ and $f'_z(\gamma_z)=f'_z(\gamma^u_z)$.
	
	This means there is some other literal  $l_y\in C$ with variable $y$ such that it is satisfied in $\gamma^u \cup f(\gamma^u)$. As $f$ and $f'$ must agree on $\gamma^u$, $l_y$ is also satisfied by $\gamma^u \cup f'(\gamma^u)$.
	$\gamma \cup f(\gamma)$ does not satisfy $l_y$ as it does not satisfy $C$. Therefore the only $\forall$ variable $y$ can be is $u$, however $l_y$ cannot be $l_u$, otherwise $\gamma \cup f(\gamma)$ would satisfy this (we come back to this for pure paths), and  $l_y$ cannot be $\lnot l_u$ otherwise we would have $\lnot l_u \not \pusim l_z$.
	Therefore $y$ is existential. Let
	\[ f''=\{f''_w \mid w \in \var_{\exists}(\Pi) \}\quad
	f''_w(\tau) = \begin{cases}
		\neg f'_w(\tau) & \text{ if } \tau = \gamma_y \text{ and } w=y\\
		f'_w (\tau)     & \text{ otherwise,}
	\end{cases}
	\]
	we show that $f''\in\mpool$:
	\begin{enumerate}
		\item (+ smaller than $f'$) As $\gamma \cup f(\gamma)$ falsifies $l_y$ and $\gamma \cup f(\gamma^u)$ satisfies $l_y$,	 	
		$(\gamma_y, y, u)$ and  $(\gamma^u_y, y, u)$ are dependency witnesses for $f'$ and therefore must be for $f$ by the induction hypothesis. $f''_y(\gamma_y)=  f''_y(\gamma^u_y)$ and thus these dependency witnesses are removed for $f''$. 
		For any dependency witness $(\delta, w, u)$ in $f''$, $w\neq y$ or ($\delta\neq \gamma_y$ and $\delta\neq \gamma_y^u$ ), and so $f''_w(\delta)=f'_w(\delta)$ and $f''_w(\delta^u)=f'_w(\delta^u)$ so $(\delta, w, u)$ is a dependency witness in $f'$ and thus also in $f$. 
		\item $f''$ only differs from $f'$ on $\gamma_y$ and its extensions.
		Since $(\gamma_y, y,u)$ is a dependency witness for $f'$, then $u\in D(y)$, and since it is consistent with $\gamma$ then $\gamma_y$ satisfies $l_u$.
		\item $f_w''$ only differs from $f'_w$ when $w=y$. 
		We know that $(\gamma_y, y,u)$ is a dependency witness for $f$, so as $f_y(\gamma^u_y)$ satisfies $l_y$ then $f_y(\gamma_y)$ satisfies $\bar l_y$. Assume for contradiction, that $\lnot l_u  \pusim \lnot l_y$. We know that if $l_u$ was in $C$ then $\gamma$ would satisfy $C$ which is does not (this is the only observation needed to get from \Drrs to \Dpu). Hence extending $\lnot l_u  \pusim \lnot l_y$ through the positive $l_y$ to $l_z$ gives us a $\lnot l_u$ pure path from $\lnot l_u$ to $l_z$ giving us a contradiction.
	\end{enumerate}
	
	Since the set of dependency witnesses is finite the induction step eventually terminates with model where $(\alpha, x, u)$ is no longer a dependency witness, and we have strictly fewer dependency witnesses on $u$.
	We can repeat this for any other witness $(\alpha', x, u)$ until all witnesses of $(\alpha', x, u)$ are removed to get model $f^*$. We can restrict the domain of $f^*_x$ to remove $u$ and it will still be well defined and a winning Skolem function.
\end{proof}

We note that a known dependency witness $(\alpha, x, u)$ can be eliminated in polynomial time (there is at most one flip per existential variable).
On the other hand, checking whether a given model $f$ has a forbidden dependency witness is $\NP$-complete.
\begin{theorem}
	Given a true (D)QBF $\Phi$ and its model $f$, it is NP-complete to decide if $f$ has a dependency witness $(\alpha, x, u)$ for universal $u$ and existential $x$ with $(u,x) \notin \Dpu(\Phi)$.
\end{theorem}
\begin{proof}
	By reduction from SAT.
	Take a CNF $\phi(V)$ in the variables $V$ and construct the true DQBF $\Phi$ with prefix $\forall V \cup \{u\} \exists x(V \cup \{u\})$ where $\{u, x\} \cap V = \emptyset$ and matrix with the single clause $\bigvee_{v \in V} v \lor u \lor \bar x$, with the Skolem function $f_x = \phi(V) \land u$.
	$f_x$ is a model, since whenever $u$ is set to false, $f_x$ yields $0$ satisfying the only clause.
	$(u,x) \notin \Dpu(\Phi)$ since the literal $x$ does not occur in $\Phi$ at all.
	Yet for any universal assignment $\alpha$, $(\alpha, u, x)$ is a dependency witness for $f_x$ if and only if $\alpha|_V$ is a satisfying assignment to $\phi$.
\end{proof}


\subsection{Strategy Extraction in Long Distance Q-Resolution}\label{sec:ld}

Theorem~\ref{thm:dpu-soundness} establishes semantic soundness of $\Dpu$: as a DQBF mapping it preserves both falsity (trivially) and truth (Theorem~\ref{thm:dpu-soundness}).
This means that $\Dpu$ can be soundly used in any DQBF proof system, and in fact the dependency scheme is not even `used in' the proof system any more, it operates entirely outside of it.
An important case not directly addressed by Theorem~\ref{thm:dpu-soundness} is that of \emph{long-distance Q-resolution (\lqrc)}~\cite{BJ12}.
\lqrc 
is a sound proof system for QBF and is supported by the solvers \depqbf~\cite{DBLP:journals/jsat/LonsingB10} and \qute~\cite{PeitlSS19}, but it is not sound for DQBF~\cite{BeyersdorffBlinkhornChewSchmidtSuda19}, so from Theorem~\ref{thm:dpu-soundness} alone we cannot infer that \lqdpurc is sound (we give the definition of long-distance Q-resolution with a dependency scheme $\mathcal{D}$ in Figure~\ref{fig:ldqdres} and this works for  $\mathcal{D}=\Dpu$; for long distance Q-resolution without a dependency scheme we take $\mathcal{D}=\Dtrv$, the trivial dependency scheme, where $(u,x)\in \Dtrv(\Pi\phi)$ if and only if $u\in \D{x}$). 
Fortunately, Theorem 8 from~\cite{BB16} addresses precisely this: it shows that \lqDrc\ is sound when $\mathcal{D}$ preserves DQBF truth value (the technical term for this used in~\cite{BB16} is `full exhibition').

\begin{theorem}
	\label{thm:ldqdpures-soundness}
	\lqdpurc is a sound and complete proof system for QBF.
\end{theorem}

\begin{figure}[ht]
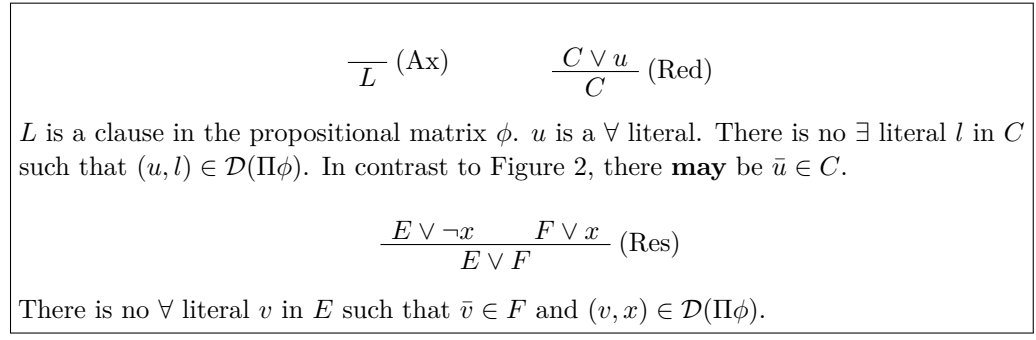

	\framebox{\parbox{0.95\textwidth}{
			\begin{prooftree}
				\AxiomC{}
				\RightLabel{(Ax)}
				\UnaryInfC{$L$}
				\DisplayProof\hspace{1cm}
				\AxiomC{$C \vee u$}
				\RightLabel{(Red)}
				\UnaryInfC{$C$}
			\end{prooftree}
			
			$L$ is a clause in the propositional matrix $\phi$. 
			$u$ is a $\forall$ literal. There is no $\exists$ literal $l$ in $C$ such that $(u,l)\in \mathcal{D}(\Pi \phi)$.
			In contrast to Figure~\ref{fig:dres}, there \textbf{may} be $\bar u\in C$.
			
			\begin{prooftree}
				\AxiomC{$E \vee \neg x$}
				\AxiomC{$F \vee x$}
				\RightLabel{(Res)}
				\BinaryInfC{$E \vee F$}
			\end{prooftree}
			There is no $\forall$ literal $v$ in $E$ such that $\bar v\in F$ and $(v,x)\in \mathcal{D}(\Pi \phi)$.
	}}
	\caption{Proof rules of \lqDrc\ applied to an input QBF $\Pi\phi$ for dependency scheme $\mathcal{D}$.
		\label{fig:ldqdres}}
\end{figure}

It can also be shown using methods from~\cite{PeitlSS19a} that $\lqdpurc$ admits polynomial-time strategy extraction.
The proof follows the proof outline of~\cite[Theorem 2]{PeitlSS19a}.
In particular, one can show that $\Dpu$ is a \emph{normal} dependency scheme.
A dependency scheme $\mathcal{D}$ is normal~\cite[Definition 7]{PeitlSS19a} if any \lqDrc refutation of a QBF with outermost universal variables contains outermost variables in at most one polarity (and additionally the proofs are closed under application of partial existential assignments in a natural way). 
This unique polarity prescribes the winning strategy for the outermost variables, and this idea can be captured in a polynomial-size circuit, yielding both soundness and strategy extraction for \lqDrc for normal $\mathcal{D}$~\cite[Theorem 1]{PeitlSS19a}.

In order to prove that $\Dpu$ is normal, we follow the recipe of~\cite[Section 5.2]{PeitlSS19a}.
As in there, we restrict ourselves to QBFs of the form $\forall u \exists x_1, \dots, \exists x_n \phi$: with a single, outermost universal variable.
It is easy to see that one can forget all literals on other universal variables without changing the validity of an \lqdpurc refutation.

\begin{lemma}
	\label{lem:learned-paths}
	If a clause $C$ is derived by \lqdpurc from a QBF $\forall u\Pi\phi$, then $(u,x) \in \Dpu(\Pi\phi) \iff (u, x) \in \Dpu(\Pi\phi \wedge C)$ for all $x$, i.e., $C$ can be soundly used to calculate $\Dpu$ dependencies of an outermost variable $u$ as if $C$ were an input clause.
\end{lemma}
\begin{proof}
	Clause addition clearly does not take away any resolution-path connections.
	We need to show that it does not create new connections either.
	
	The case of resolution steps is identical to the case of $\Drrs$~\cite[Lemma 1]{PeitlSS19a}.
	Suppose $C_1$ and $C_2$ are resolved on $x$ to obtain $C$, and a path uses the clause $C$ and its two literals $p_1, p_2$.
	If $\{p_1, p_2\} \subseteq C_i$, replace $C$ with $C_i$, and otherwise replace $C$ with $C_1, C_2$ connected by the pivot literals $x, \bar x$.
	Either way, $C$ can be replaced by $C_1$ and $C_2$ to show the same connections.
	$u$-purity is trivially preserved: if $C$ does not contain $u$ or $\bar u$, neither $C_1$ nor $C_2$ contain it.
	
	The case of reduction steps, trivial for $\Drrs$, requires some care.
	Suppose the reduction step $C \vee u \rightarrow C$ introduces pure paths.
	This means there is now a path $P$ from $\bar u$ to some $x$ that uses $C$ but could not use $C \vee u$ (because of universal impurity).
	So this path uses a literal $p \in C$ via which the clause $C$ is entered.
	Consider the prefix of $P$ that ends in $\bar p \in C'$, just before a transition from $C'$ to $C$ is made.
	This is a $\bar u$-pure path that shows $\bar u \pusim \bar p$.
	But since both $p, u \in (C \vee u)$, also $u \pusim p$, and thus $(u, p) \in \Dpu$, contradicting the soundness of the reduction step $C \vee u \to C$.
	Thus, the reduced clause $C$ is not useful for any new resolution-path connections.
\end{proof}

\begin{lemma}
	\label{lem:ldqdpures-nomerge}
	A clause $C$ derived by \lqdpurc from a formula $\forall u\Pi\phi$ cannot contain both $u$ and $\bar u$.
\end{lemma}
\begin{proof}
	Consider the resolution step of $C_1$ and $C_2$ over the pivot variable $x$ that produced the first clause $C$ with $u$ and $\bar u$ (axioms are not tautological, so such a clause must have been produced by resolution).
	Without loss of generality $x, u \in C_1, \bar x, \bar u \in C_2$.
	But then by Lemma~\ref{lem:learned-paths} $u \pusim x$ and $\bar u \pusim \bar x$.
	Thus, the resolution step is not valid in \lqdpurc.
\end{proof}

\begin{theorem}
	\label{thm:ldqdpures-strategy-extraction}
	There is a polynomial time algorithm that, given an \lqdpurc refutation of a QBF $\Pi\phi$, computes a strategy for the universal player.
\end{theorem}

\begin{proof}
	Consider an \lqdpurc refutation of a formula $\forall u \Pi \phi$ in which both $u$ and $\bar u$ occur.
	Copy clauses to make the refutation tree-like and take a minimal sub-derivation of some (not necessarily empty) clause $C$ that still contains both $u$ and $\bar u$.
	Call this sub-derivation, which ends in the clause $C$, $P$.
	By Lemma~\ref{lem:ldqdpures-nomerge}, $u$ and $\bar u$ do not occur in $C$: if $u \in C$, then omit all reduction steps on $u$, otherwise omit all reduction steps on $\bar u$.
	Call the resulting, still valid, \lqdpurc derivation $P^*$ and its final clause $C^*$, without loss of generality $u \in C^*, \bar u \not \in C^*$.
	$P^*$ must have a reduction step on $\bar u$ in the presence of some literal $x$, so that $(u, \var(x)) \not \in \Dpu$.
	Take the lowermost such reduction step $C_0 \vee \bar u \to C_0$.
	It follows that the bottom of $P^*$ is shaped as follows:
	
	\begin{prooftree}
		\AxiomC{$C_0 \vee \bar u$}
		\UnaryInfC{$C_0$}
		\AxiomC{$C_0'$}
		\BinaryInfC{$C_1$}
		\AxiomC{$C_1'$}
		\LeftLabel{$\ddots$}
		\BinaryInfC{$C_k$}
		\AxiomC{$C_k'$}
		\BinaryInfC{$C^*$}
	\end{prooftree}
	
	and no $C_i, C_i'$ contains $\bar u$.
	Let $p_i$ be the pivot for the resolution step producing $C_i$.
	There is a resolution path from $u \in C^*$, through the pivots $p_k, \dots, p_1$ to $C_0 \vee \bar u$, establishing that $\bar u \pusim p_1$ and $u \pusim \bar p_1$, a contradiction with the soundness of the reduction step.
	
	This demonstrates that $\Dpu$ is normal. 
	The rest of the proof is identical to the proof in~\cite{PeitlSS19a}.
\end{proof}

\subsection{Separations}
\label{sec:separations}

\lqdpurc trivially p-simulates \lqdrc as all \lqdrc refutations are in  fact \lqdpurc refutations already. Not all \lqdpurc proofs are \lqdrc proofs as some reduction and resolution steps could be prohibited and in fact we show that there is no workaround.

In QBF, the \qparity\xspace formulas are a family of false formulas which require the $\forall$ player to play the parity function in order to win. The parity function's hardness on bounded-depth formulas usually translates \cite{BWJ14,BCJ19} to proof size lower bounds, but in some cases gadgets can be used to find short proofs depending on the proof system. Therefore, by tuning these gadgets we can use variations of \qparity\xspace to separate different QBF proof systems.

\begin{definition}[ts-\lqparity(N)]
	Let $\xor_l(o_1,o_2,o, z)$ be the set of clauses
	$ \{
	(z\lor\neg o_1\lor\neg o_2\lor\neg o),\,
	(z\lor o_1\lor o_2\lor\neg o),\,
	(z\lor\neg o_1\lor o_2\lor o),\,
	(z\lor o_1\lor\neg o_2\lor o)
	\}$
	\begin{align*}
		\exists x_1,\dots,x_N\;\forall z\;\exists t_2,\dots,t_N,s_2,\dots,s_N.\quad
		\bigwedge \xor_l(x_1,x_2,t_2, z) 
		\wedge\;\bigwedge_{i=3}^{N}\xor_l(t_{i-1},x_i,t_i,z) \\
		\wedge\;\bigwedge \xor_l(x_1,x_2,s_2, \neg z) 
		\wedge\;\bigwedge_{i=3}^{N}\xor_l(s_{i-1},x_i,s_i,\neg z) 
		\wedge\;(z\lor t_N)\;\wedge\;(\neg z\lor\neg s_N).
	\end{align*}
\end{definition}

\begin{lemma}\label{lem:qpar_nomerge}
	The shortest \lqrc refutations of ts-\lqparity(N) are \qrc refutations.
\end{lemma}
\begin{proof}
	At the beginning of the proof, every clause contains a $z$ or $\bar z$ literal block by some inner existential literal.
	Every derived clause in \lqrc that contains an inner existential literal must therefore contain a $z, \bar z$ literal or both. 
	This should be intuitive as universal literals linger unless they can be reduced which they cannot in the presence of these inner existentials. In addition any derived clause that contains a $z, \bar z$ or both without a blocking existential literal can be immediately reduced without cost to the proof size, so $z, \bar z$ literals and inner existential literals can only coincide in reduced clauses. 
	For more details follow the argument from \cite{BCJ19}. 
	As argued in \cite{BCJ19}, long distance (merge) steps are possible, but can never be reduced as this must be done after resolving the inner existential literal.
	Resolving the inner existential literal cannot be done because it will always be an illegal merge step. 
\end{proof}

\begin{lemma}\label{lem:qpar_lb}
	The shortest \lqrc refutations of ts-\lqparity(N) are exponential in $N$.
\end{lemma}
\begin{proof}
	Citing \cite[Theorem~26]{BCJ19} we use the well established strategy extraction lower bound technique. The winning universal strategy extracted from a \qrc proof is always a bounded-depth circuit. In this case the winning strategy for $z$ is the parity function on $x_1 \dots x_n$. Parity has exponential lower bounds in bounded-depth circuits \cite{FSS84,has88}, therefore since the strategy extraction was done in polynomial time in the size of the proof, the proofs must be at least exponential size.
	The shortest \lqrc proofs are the shortest \qrc proofs, so these are exponentially bounded below as well.
\end{proof}

We can observe that actually ts-\lqparity(N) will not be a hard problem when using the proof systems of \lqdrc or \lqdpurc because $\Drrs$ and $\Dpu$ are empty and the problem reduces to the SAT benchmark \texttt{Dubois} which is a known easy family of formulas.
While we could use this as a separating example between \lqrc and \lqdpurc, we can do better and add a gadget so that it also becomes a separating family between \lqdrc and \lqdpurc.

\begin{definition}[Bridged ts-\lqparity]
	\begin{align*}
		\exists x_1,\dots,x_N\;\forall z\;\exists t_2,\dots,t_N,\;\exists s_2,\dots,s_N,\; \exists b\;
		(\text{all clauses from ts-\lqparity})\wedge\\(z \vee \neg t_N \vee b)\wedge (z \vee t_N \vee \neg b) \wedge (\neg z \vee \neg s_N \vee b)\wedge (\neg z \vee s_N \vee \neg b)
	\end{align*}
\end{definition}

The $b$ variable must be equal to $t_N$ when  $z$ is false and equal to $s_N$ when $z$ is true, and that is the only condition needed to satisfy these clause.
Making this modification does not change the proofs of Lemmas~\ref{lem:qpar_nomerge}~and~\ref{lem:qpar_lb}. But previously there were no resolution paths between $t$-variables and $s$-variables. But now $b$ acts as a bridge.

\begin{lemma}
	$\Drrs$ and $\Dtrv$ are equivalent on Bridged ts-\lqparity.
\end{lemma}

\begin{proof}
	\noindent \textbf{Induction hypothesis} (on $i$): $z\rrssim t_{N-i}$ and  $\neg z \rrssim\neg t_{N-i}$.
	
	\noindent \textbf{Base case}: $z\vee t_N$ is an axiom. 
	$(\neg z \vee  s_N \vee \neg b)$ and $(z \vee \neg t_N \vee b)$ link $\neg z$ to $\neg b$ and to $\neg t_N$.
	
	\noindent \textbf{Induction step}.: We can extend a path from $t_{N+1-i}$ to $t_{N-i}$ and from  $\bar t_{N+1-i}$ to $\bar t_{N-i}$ using the clauses of $\xor_l(\bar t_{N-i},x_{N+1-i}, t_{N+1-i}, z )$.
	
	We can symmetrically do the same induction for the $s_i$ variables. Finally, although it is not necessary for hardness, $b$ appears in an axiom with $z$ and $\neg b$ appears  with $\neg z$.
\end{proof}

\begin{corollary}
	\lqdrc requires exponential-size proofs of Bridged ts-\lqparity$_N$.
\end{corollary}

\begin{lemma}
	There are short refutations of Bridged ts-\lqparity$_N$ in \lqdpurc.
\end{lemma}

\begin{proof}
	Every clause with a $t$ variable contains a positive $z$ literal.
	Every clause with an $s$ variables contains a $\neg z$ literal. 
	Therefore, for $2 \leq i \leq N$  there are no $\Dpu$ resolution paths that go from $\neg  z$ to $t_i$, from $\neg  z$ to $\neg t_i$, from $ z$ to $s_i$ nor from $ z$ to $\neg s_i$. In fact the only dependency pair in $\Dpu$ is $(z,b)$. In the derivation, with the exception of the new clause which we will not use anyway, we can reduce all $z$ and $\lnot z$ literals immediately.
	Now we have a false existential formula with a known short resolution proof. 
	Starting with $t_2$ and $s_2$ we inductively derive $t_i \leftrightarrow s_i$ (as clauses $\neg t_i \vee s_i$ and $\neg s_i \vee t_i$). Once we reach $(\neg t_i \vee s_i)$ we can contradict this with $t_i $ and $\neg s_i$. 
\end{proof}

\begin{corollary}
	\lqdrc does not p-simulate \lqdpurc.
\end{corollary}

\section{Local Pure Literal Reduction}\label{apx:lplr}
In this section, we explain how \Dpu can be useful as a rule, by directly incorporating it into the reduction steps without altering the prefix. This approach helps us simplify the proof of Theorem~\ref{thm:revsim}.

\begin{prooftree}
	\AxiomC{$C \vee u$}
	\RightLabel{(UR)}
	\UnaryInfC{$C$ }
\end{prooftree} 
In the reduction rule, if $C$ were to contain existential variables that had $u$ in their dependency sets, then a satisfaction of $C\vee u$ may not automatically mean that $C$ is satisfied regardless of $u$. 
In order to ensure that $C$ is satisfied whenever $C\vee u$ is, the existential player must be able to costlessly commit to a strategy that determines the variables of $C$ before $u$ is reached. 
A basic example is if  $u$ is pure in all clauses of the QBF $\Pi\phi$, then the existential player is effectively able to ignore the actual play of $u$ by assuming the worst case which is $u=0$. 
This can be generalised further, we can reduce $C\vee u$ when $u$ is pure only in the $\Pi\phi$ clauses that can be affected by our switch from $C\vee u $ to $C$.
We can efficiently exclude some clauses from the set of relevant clauses by using resolution paths. Suppose there is no resolution path in the $\exists$ variables right of $u$ from $C\vee u $ to $D$, then once the game is played up to $u$ and we have current partial assignment $\alpha$, then $(C\vee u)|_\alpha $ and $D|_\alpha$ cannot be in the same \qrc refutation of $\Pi\restriction_{A'} \phi|_\alpha$ if it exists. Therefore, because the presence of any $\bar u$ literal in $D$ is syntactically irrelevant to removing $u$ from $C\vee u$, the presence of any $\bar u$ literal in $D$ is semantically irrelevant to removing $u$ from $C\vee u$. 

\begin{prooftree}
	\AxiomC{$C \vee u$}
	\RightLabel{(EUR)}
	\UnaryInfC{$C$ }
\end{prooftree} 

This is how \emph{Extended Universal Reduction} (EUR) works \cite{HSB14}. You may reduce  $C \vee u$ to $C$ as long as there is no $\mathcal{S}_u$-resolution path from $C$ to any clause $D$ with $\bar u\in D$. $\mathcal{S}_u$ here would be the set of existential variables that depend on $u$. EUR can be checked in polynomial time, but requires global knowledge of the clauses, it is therefore used in the QBF proof system QRAT \cite{HSB14} where a deletion rule allows deletion of clauses that may prevent the use of EUR. 
In future updates of \texttt{dqrat-check} we may add this rule directly, while our complexity results mean this will not increase the power of the proof system, it would mean that 
QRAT proofs can be checked by our checker .

\subsection{Definition of \texorpdfstring{\lplr}{Loc-pure-Red}}

Let $\Pi \phi $ by a DQBF and let $C\vee u$ be a clause in $\phi$ with universal literal $u$, let $\mathcal{S}_u$ be the set of variables that depend on $\var(u)$. 
In extended universal reduction (EUR) we considered $\pathc{rrs}(\psi, C\vee u, \mathcal{S}_u)$ to check that no clause in it contained $\bar u$. This prevents any $\mathcal{S}_u$ literal $l$ in $C$ from having $\bar u \rrssim \bar l$, and this was sufficient for $u$ to be reduced in $C \vee u$.

We can improve on this using pure paths. Consider $\pathcpu(\bar u,\psi, C\vee u, \mathcal{S}_u)$ and check whether it contains a clause with $\bar u$ in, if it does not then $C \vee u$ can reduce to $C$. 
We will formally define how this proof rule works.

\begin{definition}
	Local Pure Literal Reduction \lplr allows us to make the following derivation
	\begin{prooftree}
		\AxiomC{$\Pi \phi \wedge C\vee u$}
		\RightLabel{(\lplr)}
		\UnaryInfC{$\Pi \phi \wedge C$}
	\end{prooftree}
	
	Where $\Pi$ is a DQBF prefix, $\phi$ is a CNF, $C$ is a clause and $u$ is a universal literal and where $\var(u)$ is found in $\Pi$. $\mathcal{S}_u$ is the set of existential variables $x$ in $\Pi$ such that $u \in \D{x}$. For brevity, $\psi$ is defined as the full DQBF $\Pi \phi \wedge C\vee u$.
	The main side condition is that $\pathcpu(\bar u,\psi, C\vee u, \mathcal{S}_u)$ does not contain any clause that contains $\bar u$.
\end{definition}

\subsection{Soundness}

We have to demonstrate soundness which we will do in DQBF, we can prove soundness by showing a p-simulation by a sound DQBF proof system and we use \dRes for this.
While there my be a shorter proof of soundness we will need this lemma for Section~\ref{sec:pequiv}.

\begin{lemma}\label{lem:dres_psim_lplr}
	We can p-simulate the \lplr rule with \dRes, by adding new variables and clauses to DQBF.
\end{lemma}

\begin{proof}
	Let $\Pi$ be the DQBF prefix and $\phi\wedge C\vee u$ be the propositional matrix of the DQBF $\psi$. 
	Suppose we reduce from $\Pi \phi \wedge C\vee u$ to $\Pi \phi \wedge C$. Consider the $\mathcal{S}_u$-literals of $C$, where $\mathcal{S}_u$ is the set of existential variables that depend on $u$. In order to make this p-simulate, for each $x\in \mathcal{S}_u$ we will replace $x$ with $x'$ where $\D{x'}\subseteq \D{x}$.
	Essentially we will recreate $\Pi \phi \wedge C$ as $\Pi' \phi' \wedge C'$ by clause additions, where we substitute each  $x'$ for $x$ to change $\Pi$ to $\Pi'$, $\phi$ to $\phi'$ and $C$ to $C'$. Technically there is no deletion rule in \dRes, we consider ``deletion'' a persistent ignoring of the clause, thereafter.

	Recall that in the definition $\pathcpu(\bar u,\psi, \chi, \mathcal{S}_u)$, we had some subset of clauses $\chi$ which we used as a start point.
	We will study pure path reachability from two different start points, firstly the singular clause set $\{C\vee u\}$. Once we have found all reachable clauses from $\{C\vee u\}$, those remaining unreachable clauses form the second start point. This way we have two ``spheres'' of reachable clauses that cover the entire set of clauses. 
	Counter-intuitively, these spheres are not necessarily disjoint because reachability requires us not to immediately re-use literals on the resolution path. The intersection of these spheres will be an important special case that we must handle. 
	Let $L^0=\pathlpu(\bar u,\psi, C\vee u, \mathcal{S}_u)$, $\chi_0 = \pathcpu(\bar u,\psi, C\vee u, \mathcal{S}_u)$. Then let $L^1=\pathlpu(\bar u, \psi, \phi\setminus \chi_0 , \mathcal{S}_u)$, $\chi_1 = \pathcpu(\bar u,\psi, \phi\setminus \chi_0, \mathcal{S}_u)$.
	We introduce conditional definition  $\bar u\rightarrow(x^{\bar u}\leftrightarrow x)$ using the IndExt rule for each variable that has a literal in $L^0$. 
	
	For each clause $D$ in $\chi_0$ we replace each literal $x\in \mathcal{S}_u$ with $x^{\bar u} \vee u$ via resolving with the definition clauses. Note that $\chi_0$ contains no $\bar u $ literals by the side condition of the rule, and no other clause except $C \vee u$ contains a $u$ literal by the path purity. Removing all $\mathcal{S}$-literals, means we can reduce any $u$-literals from what was once $\chi_0$, including the one originating from $C$ and those introduced from resolving away the $\mathcal{S}_u$-literals. We call this set $\chi^{\bar u}_0$ because it is what you would get in the $\bar u$ expansion of $\chi_0$.
	
	Now consider the literals $x$ in $L_0$ such that $x\in L_1$. In $\chi^{\bar u}_0$ we weaken all $x^{\bar u}$ literals to $x^{\bar u}\vee x$ and in $\chi_1$ we weaken all $x$ literals to $x^{\bar u}\vee x$. $\bar x$ can only appear in clauses in the intersection $\chi_0 \cap \chi_1$, because if $x$ in $L_0$ all clauses with $\bar x$ are in $\chi_0$, and likewise with $\chi_1$.  Furthermore,	$\bar x$ cannot be in $L_0$ otherwise there is a $u$-free path from $C$ to $\bar x$ and a $u$-free path from  $x$ to some $D\in \phi\setminus \chi_0$ meaning $D$ is actually in $\chi_0$.
	Likewise, $\bar x$ cannot be in $L_1$ otherwise there is a $u$-free path from $C$ to $ x$ and a $u$-free path from  $\bar x$ to some $D\in \phi\setminus \chi_0$ meaning $D$ is actually in $\chi_0$.
	
	Each clause $D\in \chi_0 \cap \chi_1$ has a unique entry literal $\bar x$ which cannot be in $L_0$ nor $L_1$. This means its sufficient just to have two copies, one originating from $\chi_0$ and one originating from $\chi_1$. Note that \dRes p-simulates Frege rules because it is p-equivalent to \dFregeRed, so after the weakening we can use distributivity on the two copies of the intersection clause, to replace $\bar x$ with  $\bar x^{\bar u} \wedge \bar x$. 
	
	The full replacement scheme is as follows
	
	$$x'= \begin{cases}
		x^{\bar u}  & \text{if } x \in L_0, x \notin L_1,\\
		x^{\bar u}\vee x  & \text{if } x \in L_0, x \in L_1,\\
		x^{\bar u}\wedge x  & \text{if } \bar x \in L_0, \bar x \in L_1,\\
		x & \text{otherwise.}
	\end{cases}$$
	
	$x'$ has dependency set $\D{x'} \subseteq \D{x}$, and replaces $x$ in all clauses, all extra literals introduced have been reduced.
	
	Therefore we have derived a set of clauses under a set of variables that have the same structure as the original set such that the new set of clauses is the correct substitution of variables of the original set of clauses, with the exception that the substitution of $C\vee u$ is replaced by a substitution of $C$.
	Despite \dRes not having a deletion rule, we can ignore all variables replaced by substitution and all original clauses and clauses used as in intermediate part of of this proof. This is practically deletion.
	The only difference is that $\D{x'}$ may be \textit{strictly} smaller than $\D{x}$, but this does not prevent any future steps. One can easily imagine \dRes with a dependency weakening rule. 
\end{proof}

For valid (\lplr) steps, note that a checker should, in theory, require no more computation than an simple EUR checker, because the resolution paths are shorter.
Every valid EUR step is also an \lplr step, so a QRAT checker (such as \qrattrim) need only check for \lplr and not EUR. 

The advantage of using \lplr over EUR is that when p-simulating expansion based solving \cite{HSB14,KS19}, one will not have to delete the definitions $\alpha \rightarrow (x= x^\alpha)$ in order to remove universal literals.



\section{P-Equivalence with \texorpdfstring{\dRes}{IndExtQURes}} \label{sec:pequiv}

When looking at proof complexity, \dRes\cite{ChewPeitl25} has been proven to be very powerful relative to the other proof systems for both QBF and DQBF. \dRes has been shown to p-simulate the majority of QBF and DQBF proof systems (see Fig~\ref{fig:QBFsimstruct} for QBF proof systems), despite \dRes only using a small number of simple rules. The p-simulation of many QBF and DQBF techniques presents an opportunity to improve certification for both logics.

\begin{figure}[ht]
	\centering{
		\begin{tikzpicture}[xscale=1.2]
			\node[cdclcalcn](DFred) at (0,10){\dRes};
			\node[strongcalcn](G) at (-2.5,9){\Gfull};
			\node[strongcalcn](G1) at (-2.5,8.25){\textsf{G}$_1$};
			\node[strongcalcn](G1t) at (-1.35,7.9125){\textsf{G}$^*_1$};
			\node[strongcalcn](QRATpu) at (0,9){\textsf{QRAT(\lplr)}};
			\node[strongcalcn](QRAT) at (0,8.25){\textsf{QRAT(EUR)}};
			\node[expcalcn](EFred) at (0, 7){\textsf{QRAT(UR)}};
			\node[strongcalcn](Fexp) at (3.25, 7.825){\textsf{Frege}+$\forall$\textsf{Exp}};
			\node[strongcalcn](idrc) at (4.25, 8.825){\idrc};
			\node[expcalcn](Fred) at (-2.25,6){\FregeRed};
			\node[expcalcn](Cred) at (-2.25,5){$\mathsf{AC}_0$\FregeRed};
			\node[expcalcn](LQM) at (-4.5,3){\lqrc$\setminus\{\Red\}$};
			\node[expcalcn](M) at (-4.5,5){\textsf{M-Res}};
			\node[expcalcn](QU) at (-2.25,4){\qurc};
			\node[expcalcn](Q) at (0,3){\qrc};
			\node[expcalcn](LQ) at (0, 4){\lqrc};
			\node[expcalcn](LQU) at (0, 5){\lqurc};
			\node[expcalcn](LQUP) at (0, 6){\lquprc};
			\node[expcalcn](QS) at (2, 4){\qsrc};
			\node[expcalcn](QD) at (2, 5){\qdrc};
			\node[expcalcn](LQD) at (2, 6){\lqdrc};
			\node[expcalcn](LQDp) at (2,7){\bm{$\mathsf{LDQ}(\mathcal{D}^{\forall\mathsf{pure}})$}};
			\node[expcalcn](QCDCL) at (-2.25, 3){\texttt{QCDCL}};
			\node[expcalcn](e) at (4,3){\ecalculus};
			\node[expcalcn](ir) at (4,4){\irc};
			\node[expcalcn](irm) at (4,5){\irmc};
			
			\draw[dashed, forestgreen](-6,7.5)--(5.1,7.5);
			\draw[->,forestgreen](-5.75,7.45)--(-5.75,7);
			\node(selabel) at (-3.75,7.25){\textcolor{forestgreen}{Known Strategy Extraction}};
			\draw(Q)--(QU)--(Cred)--(Fred);
			\draw(LQUP)--(EFred)--(QRAT);
			\draw(e)--(ir)--(irm);
			\draw(DFred)--(QRATpu)--(QRAT);
			\draw(QRAT)--(Fexp);
			\draw(idrc)--(4.25, 5.75) arc (0:90:-0.5\radius)--(4.75, 5.5) arc (270:180:-0.5\radius)--(5, 4.75) arc (0:-90:0.5\radius)--(4.25,4.5) arc (90:180: 0.5 \radius);
			\draw(Fexp)--(3.25, 3.5);
			
			\draw (3.25, 3.5) arc(0:90:-0.25\radius) (3.35,3.375) ;
			\draw (3.35,3.375)--(3.75,3.375);
			\draw (3.70,3.375) arc(-90:-180:-0.25\radius) (e) ;
			\draw(QD)--(LQD);
			\draw(QS)--(QD);
			\draw(1.75,3.5) arc(-90:0:0.5\radius) (QS);
			
			\draw(LQ)--(LQU)--(LQUP);

			\draw(0.25,3.5)--(3.75,3.5);
			\draw(3.75,3.5) arc(-90:0:0.5\radius) (ir);
			\draw(-0.25,3.5) arc(-90:0:0.5\radius) (ir);
			
			\draw(0.5,4.5)--(3.75,4.5);
			\draw(0.5,4.5) arc(90:180:0.5\radius) (LQ);
			\draw(3.75,4.5) arc(-90:0:0.5\radius) (irm);
			
			\draw(-2,6.5)--(-0.25,6.5);
			\draw(-2,6.5) arc(90:180:0.5\radius) (LQUP);
			\draw(-0.25,6.5) arc(-90:0:0.5\radius)-- (EFred);
			\draw(-2,3.5) arc(90:180:0.5\radius);
			
			\draw(Q)--(LQ);

			\draw(3.75,6.5)--(0.25,6.5);
			\draw(3.75,6.5) arc(90:0:0.5\radius) --(irm);
			\draw(1.75,6.5) arc(90:0:0.5\radius) --(LQD);
			\draw(0.25,6.5) arc(270:180:0.5\radius)-- (EFred);
			
			\draw(-4.25,6.5)--(-0.25,6.5);
			\draw(-4.25,6.5) arc(90:180:0.5\radius)--(M)--(LQM);
			\draw(-0.25,6.5) arc(-90:0:0.5\radius)-- (EFred);

			\draw(-0.25,3.5)--(-4.25,3.5);
			\draw(-4.5,3.25) arc(180:90:0.5\radius) (LQ);
			\draw(-2.25,9.5) arc(90:180:0.5\radius)--(G);
			\draw(-2.25,9.5)--(-0.25,9.5);
			\draw(-0.25,9.5) arc(-90:0:0.5\radius)--(DFred);
			\draw(0.25,9.5) arc(270:180:0.5\radius)-- (DFred);
			\draw(4, 9.5) arc(90:0:0.5\radius)--(idrc);
			\draw(0.25,9.5)--(4, 9.5);
			\draw(QU)--(LQU);
			
			\draw (0.625,4.5) arc(90:180:-0.5\radius)--(0.875,4.75)--(0.875,5.25) arc(180:90:0.5\radius)--(1.25,5.5)--(1.75,5.5) arc(90:180:-0.5\radius)--(LQD);
			
			\draw(LQD)--(LQDp);
			
			
			\draw(G1)--(G);
			\draw(EFred)--(G1t)--(G1);
			\draw(QRAT)--(G);

			\draw(0.25,6.5) arc(270:180:0.5\radius)-- (EFred);

			\draw(0.25,3.5)--(0.75,3.5);
			\draw(0.25,3.5) arc(90:180:0.5\radius) (Q);

		\end{tikzpicture}
	
	}
	\caption{The p-simulation structure of refutational QBF proof systems \cite{BJ12,BWJ14,BBM18,Bohm21,BBCP20,BCJ19,Chew25round,CH22,ChewS24,ChewPeitl25,HSB14,KS19,KBKF95,KP90,PeitlSS19a,rabe17,Slivovsky-sat14,Gelder11}. \label{fig:QBFsimstruct}
	}
\end{figure}

One way to do this is to adapt an existing certification format with new rules so that it p-simulates \dRes, therefore transitively p-simulating most of the techniques in QBF and DQBF. Our main idea is that we can combine \dqrat with the \Dpu prefix modification rule. Doing so gives a proof system that we will show is p-equivalent to \dRes. 
Recall the definition of \dqrat (Figure~\ref{fig:dqrat}), we introduce \Dpu as a replacement of the \Drrs rule:
$$\frac{\Pi \phi}{\Omega \phi} (\Dpu)$$

Where $\Pi$ and $\Omega$ are DQBF prefixes and $\phi$ is a CNF.
The condition on $\Omega$ is that it contains the same variables as $\Pi$, with a modification of the dependency sets with the restriction that $u\notin \mathrm{D}_{x}^\Omega$ only if $u\notin \D{x}$ or $(u,x) \notin \Dpu(\Pi\phi)$.

\begin{definition}
	(The refutational version of) \dqratpu is a proof system that allows proofs where each line is an S-form DQBF $\Pi \phi$. Refutation is shown in the same way as \dqrat. Each subsequent line follows from the previous by one of the seven rules: ATA, Del, UR, DQRAT$_\exists$, DQRAT$_\forall$, BPM and \Dpu. 
\end{definition}

First we show that \dqratpu p-simulates \dRes. This direction is the more important of the two if we want to use \dqratpu for certification.

\begin{theorem}
	\dqratpu p-simulates \dRes.
\end{theorem}
\begin{proof}
	
We consider an \dRes refutation $\pi$ of $\Pi \phi$ as a sequence of clauses $C_1 \dots C_n$ with $C_n=\bot$.
We will p-simulate $\pi$ by creating a \dqratpu derivation $L_1 \dots L_m$. Within that sequence there will be a subsequence $(L'_{f(i)})_{i=1}^{n}$ where $L'_{f(i)}= \Omega \psi$ and $\psi$ contains clauses $C_1 \dots C_i$ as well as any clauses from $\phi$ and $\Omega$ quantifies all variables appearing in $\psi$, with the same quantifiers and dependency sets as in the \dRes proof.

For the \textbf{(Ax)} rule we ensure that we keep all clauses from $\phi$ in $\psi$.
For \textbf{adding variables to the prefix} we can use BPM.  
For the \textbf{(Red)} rule if we want to use clause $C \vee u$ to get $C$ we use the UR rule in \dqrat. Technically we need to keep a copy of $C\vee u$ in $\psi$ which can be returned by using ATA.
For the \textbf{(Res)} rule we can use ATA to add any resolvent since if $D_1\vee \bar x$ and $D_2\vee x$ are in $\psi$ then $\psi \wedge \neg D_1 \wedge \neg D_2$ is a propositional contradiction, furthermore we can derive it via reverse unit propagation as the units of $\neg D_1$ simplify $D_1\vee \bar x$ to just $\bar x$ and the units of $\neg D_2$ simplify $D_2\vee x$ to just $x$, we then propagate to the empty clause. 

Suppose we add independent extension clauses $ (\bar \alpha \vee n \vee a), (\bar \alpha \vee n \vee b), (\bar \alpha \vee \bar n \vee \bar a \vee \bar b)$. $\D{n}= \D{a}\cup \D{b} \setminus \D{\alpha}$ using the \textbf{(IndExt)} rule.
We first choose to add the existential variable $n$ with $\D{n}= \D{a}\cup \D{b}$ using BPM. 
Using DQRAT$_\exists$ we can add the first two clauses $(\bar \alpha \vee n \vee a), (\bar \alpha \vee n \vee b)$, provided $n$ is a new variable. In order to add the final clause, we need the outer variables of $(\bar \alpha \vee n \vee a), (\bar \alpha \vee n \vee b)$ to each have some non $ n$ literal to be opposite of a literal in $(\bar \alpha \vee \bar n \vee \bar a \vee \bar b)$, this can only be $a$ and $b$. So in this case because we chose that $\D{n}= \D{a}\cup \D{b}$, this is sufficient to add the final clause  $(\bar \alpha \vee \bar n \vee \bar a \vee \bar b)$.

Finally we need to remove the dependencies of $n$ that are in $\alpha$. 
Suppose we have a universal variable $u$ and a literal $l_u$ is a conjunct in the assignment $\alpha$.
All paths from $l_u$ to $n$ or $\bar n $ pass through $\bar l_u$
as our definition clauses are the only clauses that contain $n$ and $\bar n$ at this point.
Thus $(u,n)\notin \Dpu$. 
Therefore we can drop $\var(u)$ from \D{n}. We can do this for each literal in $\alpha$. At the end $\D{n}= \D{a}\cup \D{b} \setminus \D{\alpha}$.
\end{proof}

The reverse is true, that \dRes p-simulates \dqratpu. For the original rules of \dqrat we already know how to do a p-simulation \cite{ChewPeitl25}. Only the \Dpu rule remains, however this ends up being similar enough to the p-simulation of the \Drrs rule. 

\newcounter{savedthm}
\setcounter{savedthm}{\value{theorem}}
\begin{theorem}\label{thm:revsim}
	\dRes p-simulates \dqratpu.
\end{theorem}
\begin{proof}
	We know \dRes p-simulates \dqrat already. What we have to show is that \dRes p-simulates the \Dpu rule. To do this we follow a similar proof to \dRes p-simulating the \Drrs rule, for each spurious dependency $(u,x)$ we replace $x$ with another variable $x'$ where $D_{x'}=D_x \setminus \{u\}$. 
	
	Given a DQBF with propositional CNF matrix $\phi$ and let $u$ be some universal variable in the prefix, and $x$ be an existential variable such that x depends on $u$, but only spuriously: $u\in \D{x}, (u,x)\notin \Dpu$.	We define $\chi_u$ to be the subset of $\phi$ where all clauses contain literal $u$ and $\chi_{\bar u}$ to be the subset of $\phi$ where all clauses contain literal $\bar u$.
	Define $L^u=\pathlpu(u,\psi, \chi_u, \mathcal{S})$, $L^{\bar u}=\pathlpu(\bar u,\psi, \chi_{\bar u}, \mathcal{S}_u)$, where $\mathcal{S}_u$ is the set of existential variables that contain $u$ in its dependency set. 
	
	If $(u,x)\notin \Dpu(\Pi \phi)$ then we have four cases (up to symmetry) of $x$'s membership in $L^u$ and $L^{\bar u}$:
	
	\begin{minipage}{0.95\textwidth}
		\begin{multicols}{2}
			\begin{enumerate}
				\item $x\notin L_{u}$, $\bar x\notin L_{u}$, $x\notin L_{\bar u}$ and $\bar x\notin L_{ \bar u}$
				\item $x\in L_{u}$, $\bar x\notin L_{u}$, $x\notin L_{\bar u}$ and $\bar x\notin L_{ \bar u}$
				\item $x\in L_{u}$, $\bar x\notin L_{u}$, $x\in L_{\bar u}$ and $\bar x\notin L_{ \bar u}$
				\item $x\in L_{u}$, $\bar x\in L_{u}$, $x\notin L_{\bar u}$ and $\bar x\notin L_{ \bar u}$
			\end{enumerate}
		\end{multicols}
	\end{minipage}
	\medskip
	
	For cases 1, 2 and 3, we can take advantage of the fact that $\bar x$ has no pure path to either $u$ or $\bar u$.
	In fact there is no path at all from $\bar x$ to either $u$ or $\bar u$, as a non-pure path could be minimised to a pure path of one of the polarities. 
	We replace $x$ with $x^{\bar u}\vee x^{ u} $. Variables $x^{\bar u}, x^{u}$ are defined by the conditional definitions $\bar u \rightarrow (x^{\bar u} \leftrightarrow x)$, $u \rightarrow (x^{u} \leftrightarrow x)$. 
	We have to replace each $x$ literal with $x^{\bar u}\vee x^{ u} $, this is quite straightforward as we can derive $\bar x \vee x^{\bar u}\vee x^{u} $ by resolving over $u$ in the definition clauses.
	Replacing $\bar x$ is a more difficult matter. For a clause $C$ we can replace $\bar x$ with $\bar x^{\bar u} \vee u$ but we would require a way to remove the $u$ literal.
	We can remove the $u$ once we have replaced all variables in $\mathcal{S}_u$ that have spurious dependencies on $u$. In this case there cannot be any $l\in C$ such that there is a resolution path from $\bar u$ to $\bar l$ , otherwise there would be a path from $\bar u$ to $\bar x$, nor can there be a path from  $\bar l$ to $u$ for the same reason. Hence $l$ is also a case 1, 2 or 3 literal. This means we can reduce the $u$ we introduce to clause $C$, eventually after $l$ is replaced by $l^{\bar u}\vee l^{u}$, for appropriate definitions of $l^{\bar u}$ and . Thus we can replace  $\bar x$ with $\bar x^{\bar u}$ and symmetrically we can do the same to create another copy where we replace $\bar x$ with $\bar x^{ u}$. 
	By using distributivity we get $\bar x$ replaced by $\bar x^{\bar u}\wedge \bar x^{ u} $, which fortunately is the negation of what we replaced $x$ with. Thus we have $x'= x^{\bar u}\vee x^{ u} $
	
	The remaining case 4, requires us to replace $x$ with $x^{\bar u}$. Initially via resolution we can replace $x$ with $x^{\bar u}\vee u $ in some clause $C$. In the case that $u$ already was in $C$ we do not need to remove it. Otherwise we argue there are no other literals $l$ in $C$ such that there is a pure path from $\bar l$ to $\bar u$, because  there would be a pure path from $\bar u$ to $x$. This means we can use local pure literal elimination to remove said $u$, because there is no $u$-free path from $C$ to $\bar u$. Local pure literal elimination is not a rule of \dRes, however we can p-simulate the steps (Lemma~\ref{lem:dres_psim_lplr}), to perform this change. The fact we have no deletion rule in \dRes is not a factor when simulating the global conditions of \lplr, as we still select the subcnf we wish to make a copy of.

\end{proof}
\section{Practical Implications}

We took existing code from the QBF solver \qute~\cite{PeitlSS19} that detected \Drrs and naturally extended it to \Dpu. 
Using this we developed in C++ the prototype DQRAT+\Dpu checker called \dqratchk and we increased the capabilities of \qute
to include \Dpu.%
\footnote{Available at \url{https://github.com/peitl/dqrat-check} and \url{https://github.com/fslivovsky/qute}.}

\subsection{DQRAT+\texorpdfstring{\Dpu}{Dpure}-check} \label{sec:prac_check}


\texttt{dqrat-check} accepts a \texttt{.dqdimacs} file as its first argument containing the DQBF or QBF, the second argument is a proof written in the \texttt{.dqrat} format. 
We designed the \texttt{dqrat} format to follow the \texttt{qrat} format as closely as possible.
The only deviation is in how extension variables are introduced: in \texttt{qrat} they must be existential, and are quantified automatically based on other literals in the clause.
In \texttt{dqrat}, the proof can explicitly declare any existential variable and its dependencies before used in a clause (if undeclared, we adopt the QRAT convention of taking all universal variables as dependencies).

No alphabetical symbols indicates clause addition using ATA or DQRAT$_\exists$.
The \texttt{a} indicates a new $\forall$ variable.
Unlike in \texttt{.dqdimacs}, \texttt{d} indicates clause deletion (for QRAT compatibility), which leaves \texttt{e} for introduction or modification of an existential variable. Each \texttt{e} line takes multiple integers, the first is the variable, and then subsequently are the $\forall$ variables to be added or removed from the dependency set. Negative integers indicate removal.
It is possible to make an $\exists$ variable depend another $\exists$, so it will inherit the entire dependency set.

Like in QRAT, a line starting with \texttt{u} indicates a UR or DQRAT$_\forall$ rule, in each case reducing the first literal in the clause, which must be universal.

\begin{example}\label{ex:dqratcheck}
	We can write the problem from Example~\ref{ex:dres proof} in the
	\textsc{DQDIMACS} format as follows:
	
\begin{lstlisting}
c variable names are
c u 1
c v 2
c x 3
c y 4
p cnf 4 6
a 1 2 0
d 3 1 0
d 4 2 0
1 3 4 0
-1 -2 3 4 0
-1 2 3 -4 0
1 -3 -4 0
1 3 4 0
-1 -2 -3 -4 0
-1 2 -3 4 0
\end{lstlisting}
	
	The $\dRes$ proof can be written as a $\dqratpu$ proof in the following way:
	
\begin{lstlisting}
e 5 1 0
5 -1 3 0
-5 -1 -3 0
e 5 -1 0
-1 -2 5 -4 0
u -1 -2 5 -4 0
-1 2 5 4 0
u -1 2 5 4 0
-1 -2 -5 4 0
u -1 -2 -5 4 0
-1 2 -5 -4 0
u -1 2 -5 -4 0
1 -2 -5 -3 0
u -2 1 -5 -3 0
1 2 -5 3 0
u 2 1 -5 3 0
1 -5 0
u 1 -5 0
1 2 -3 0
u 2 1 -3 0
1 -2 0 
u 1 -2 0
u -2 0
\end{lstlisting}
	
	Most of the lines correspond to clauses derived using clause rules.
	The line
	\texttt{e 5 1 0}
	introduces a new existential variable~$5$ with dependency set $\{1\}$.
	After adding the two definition clauses, the dependency on~$1$ can be
	removed using
	\texttt{e 5 -1 0}.
\end{example}


We evaluated \dqratchk on formulas from this paper, as well as on proofs obtained via \qute on QBFEval 2022 benchmarks.
To obtain DQRAT proofs out of \qute, we extended the \lqrc~checker \qrup%
\footnote{Available at \url{https://github.com/peitl/qrp2rup}.}
~\cite{PeitlSS18} to convert the \lqrc~proof produced by Qute to (D)QRAT.%
\footnote{The proof is de facto QRAT, but in DQRAT format with explicit dependency sets for extension variables.}
Proof checking (\qrup + \dqratchk) took 9\% longer than solving on average (geometric mean over all $43$ false instances that \qute could solve).
Notably, \dqratchk alone accounts for only 3\% of the overall proof checking effort; 97\% is proof extraction from the \qrp trace and conversion into DQRAT, both done by \qrup.

This QBF-based scenario does not cover the full scope of what DQRAT can do; it does not use \Dpu and DQRAT$_\forall$ (not needed for our translation from \lqrc to DQRAT).
We further tested \dqratchk on crafted formulas: the formula from Example~\ref{ex:dres proof}, the bridged ts-LQParity formulas from Section~\ref{sec:separations}, and other crafted formulas from QBF literature (Equality~\cite{BBH19} and QParity~\cite{BCJ15,BCJ19}).
We generated proofs for these formulas manually (without a solver), with the exception of bridged ts-LQParity, which we solved with \qute with \Dpu and modified the prefix in such a way that the extraction through \qrup (which only works for \lqrc, not \lqdpurc) goes through.
This is possible, since the result of applying \Dpu to bridged ts-LQParity is again a QBF, and so the \lqdpurc proof can be interpreted as a sound \lqrc proof of another QBF (this holds for any QBF with a single universal variable).
\dqratchk correctly verified all of these proofs; details and instruction for reproducing are in the supplementary material.

More experiments will be interesting to evaluate the performance of \dqratchk in different scenarios, but for them to be practical, it is first necessary to develop certifying DQBF solvers that output proofs, ideally taking advantage of all DQRAT rules.
\subsection{Qute + \texorpdfstring{\Dpu}{Dpure}} \label{sec:prac_solve}

While we have motivated this work on proof checking, we were interested in the solving capabilities of \Dpu. To this end we integrated \Dpu into the dependency learning solver \qute. As with \Drrs, it is unknown whether \qute can soundly use \Dpu for term learning (to prove truth of formulas), so we only turn on \Dpu for clause learning.

We found \Dpu detects no independence on the newest QBFEval 2022 benchmarks, and \qute performs better without the overhead of \Drrs or \Dpu.
On QBFEval 2020 benchmarks
\Dpu shows promise: compared to \Drrs, $7\times$ less time is spent on computing dependencies, $2\times$ more independent pairs are found, and $3\times$ more $\Dpu$ reductions are executed (geometric means), but this did not yield more solved instances or better PAR2 score.
For ts-LQParity formulas which separate \Drrs from \Dpu \qute+\Dpu performs well as expected.

\section{Conclusion}
We have shown that \dqratpu is p-equivalent to \dRes, and thus a number of useful p-simulations via transitivity. What we cannot show by transitivity is a p-simulation of the newly created \lqdpurc, because long-distance resolution steps do not have a sound meaning in DQBF. Nonetheless we conjecture such a p-simulation will exist and we point to the literature \cite{KHS17,Chew25round,PeitlSS18} that demonstrate that there are a variety of successful of formalisations and p-simulations of long-distance resolution.
We aim to further improve and explore the capabilities of our prototype proof checker \dqratchk.
\bibliography{bib2doi}

\appendix
\section{Appendix}

\subsection{Further Details on Outer Variables in DQBF}\label{apx:outer}

Here we expand on the description of outer variables from Section~\ref{sec:outer}.

\begin{theorem}
	Let $\Pi$ be a DQBF prefix. Let $x$ and $y$ be variables in $\Pi$. If $y\in \Outer{x}$ then $\Outer{y}\subseteq \Outer{x}$.
\end{theorem}
\begin{proof}
	We will assume throughout the proof that $y\in\Outer{x}$
	
	First suppose $x$ is existential and $y$ also existential. Then $\D{y} \subseteq \D{x}$. If $z\in \Outer{y}$ then $\D{z}\subset \D{y}$ and thus $\D{z}\subset \D{x}$ meaning $z\in \Outer{x}$.
	
	Now suppose $x$ is existential and $y$ is universal, then $y \in \D{x}$. Now consider, the kernel,  $\Kernel{y}$ the set of universal variables in $\Outer{y}$. Each of these also appears in every dependency set $y$ is in including $\D{x}$, hence they also appear in $\Outer{x}$. Now suppose 	$z\in \Outer{y}$ and $z$ is existential, then $\D{z}$ only contains variables in $\Kernel{y}$ thus $\D{z}\subseteq\D{x}$ and $z\in \Outer{x}$.
	
	Next suppose $x$ is universal and $y$ is universal. We say $y$ is in $\Kernel{x}$, meaning $y$ is in all the dependency sets that $x$ is in. If $z\in \Outer{y}$ then this is defined based on all the dependency sets $y$ is in, which contain all the dependency sets of $x$. So if $z$ is in the kernel of $y$ it is in the kernel of $x$, and if $\D{x}$ contains only variables of $\Kernel{y}$ it contains only variables of $\Kernel{x}$. Thus, whether $z$ is universal or existential it is contained in $\Outer{x}$.
	
	Finally suppose $x$ is universal and $y$ is existential. Then $D_y$ only contains variables from $\Kernel{x}$. Let $z$ be such that $z\in \Outer{y}$. If $z$ is existential then $\D{z}\subseteq \D{y}$ which contains only variables from $\Kernel{x}$. Neither contain $x$ so $z\in \Outer{x}$. If $z$ is universal then $z$ is in the dependency set of $y$, and thus in the kernel of $x$.

\end{proof}

\begin{corollary}\label{cor:pr_ord}
	For a DQBF $\Pi \phi$, $ \lesssim_\Pi$ is a pre-order.
\end{corollary}

\subsection{Practical Details} \label{apx:prac}

We evaluated \qute with \Dpu on the formulas from Appendix~\ref{sec:separations}.
Even with \Drrs, \qute's running time scales exponentially, while with \Dpu the formulas are solved instantly (Figure~\ref{fig:qute-bridged-ts-lqparity}).
This is in line with proof complexity; we note that such clean mirroring of proof complexity in solver performance is far from given: previous work on QCDCL proof complexity found this business to be tricky with \qute unable to solve theoretically easy formulas quickly~\cite{BohmPB24a,BohmPB24}.
We also evaluated \qute on the PCNF track of QBFEval 2022, but saw no improvement over \Drrs or vanilla \qute.


\subsubsection{Details of the Solver \texorpdfstring{\qute}{Qute }}\label{apx:qute-depscheme-impl}

One may think that \Dpu not only finds more independence than \Drrs, but it also makes calculations faster, the latter because resolution paths that would need to be explored for \Drrs might be aborted early due to universal impurity.
This is not quite as simple.

For both $\mathcal{D} \in \{ \Drrs, \Dpu \}$, one can compute all $x$ with $(u,x) \in \mathcal{D}$ in linear time, and so all dependent pairs in overall quadratic time.
For \Drrs it is additionally possible to compute all dependencies \emph{of} some existential variable $x$ in quasilinear time with a Dijkstra-style algorithm~\cite{PeitlSS19b}.
The latter seems hard for \Dpu as resolution paths from a fixed existential literal to many universal targets may be polluted with different subsets of universal literals, leading to an exponential blow-up.
This affects the implementation of \Drrs and \Dpu in \qute.

\qute uses QCDCL (quantified conflict-driven clause learning) with \emph{dependency learning} to learn dependencies between variables dynamically.
In dependency learning, the solver starts branching and propagating as if there were no dependencies (equivalently, implicitly reducing all universal literals).
Only if this leads to a \emph{dependency conflict}, which is a resolution step in \lqrc that should have been valid but is not because of forbidden $v$ literals from Figure~\ref{fig:ldqdres}, does the solver learn missing dependencies to prevent the conflict from taking place again.
At this point a dependency scheme $\mathcal{D}$ can be inserted: if all blocking literals $v$ are found independent, the resolution step can soundly be carried out in \lqDrc.
In order to avoid a quadratic blow-up from computing upfront and storing all \Drrs dependencies, \qute only computes the required dependencies on demand during dependency conflicts (and then stores them forever).
Because multiple universal variables $v$ may be blocking in a dependency conflict, \qute computes all dependencies of the pivot variable $x$.
As shown in~\cite{PeitlSS19b}, all \Drrs dependencies of an existential variable can be found in quasilinear time with a Dijkstra-style algorithm.
The same, however, does not seem possible for \Dpu: in the search for resolution paths, any currently explored path can be impure for any subset of universal literals, and thus the actual number of paths to explore is in general exponential.
For this reason, our implementation of \Dpu computes all existential variables that depend on a given universal variable $u$, which can be done in linear time with depth or breadth-first search.
The price we pay (compared to \qute's \Drrs implementation) is that in a dependency conflict we need to compute dependencies on every blocker $v$.

\subsubsection{Experimental details}
\label{apx:experiments}

All timing is on a cluster of machines with 2x 24-core AMD EPYC 7402@2.80GHz CPUs and 1TB RAM per node, running Ubuntu 18.04.6 LTS on GNU/Linux 4.15.0-213-generic.

We evaluated \qute with \Dpu on the formulas from Appendix~\ref{sec:separations}.
Even with \Drrs, \qute's running time scales exponentially, while with \Dpu the formulas are solved instantly (Figure~\ref{fig:qute-bridged-ts-lqparity}).
This is in line with proof complexity; we note that such clean mirroring of proof complexity in solver performance is far from given: previous work on QCDCL proof complexity found this business to be tricky with \qute unable to solve theoretically easy formulas quickly~\cite{BohmPB24a,BohmPB24}.


\begin{figure}
\begin{tikzpicture}

\definecolor{darkgray176}{RGB}{176,176,176}
\definecolor{darkorange25512714}{RGB}{255,127,14}
\definecolor{forestgreen4416044}{RGB}{44,160,44}
\definecolor{lightgray204}{RGB}{204,204,204}
\definecolor{steelblue31119180}{RGB}{31,119,180}

\begin{axis}[
height=3cm,
legend cell align={left},
legend columns=3,
legend style={
  fill opacity=0.8,
  draw opacity=1,
  text opacity=1,
  at={(0.03,0.97)},
  anchor=north west,
  draw=lightgray204
},
tick align=outside,
tick pos=left,
unbounded coords=jump,
width=10cm,
x grid style={darkgray176},
xmin=-0.5, xmax=10.5,
xtick style={color=black},
xtick={0,1,2,3,4,5,6,7,8,9,10},
xtick={0,1,2,3,4,5,6,7,8,9,10},
xtick={0,1,2,3,4,5,6,7,8,9,10},
xtick={0,1,2,3,4,5,6,7,8,9,10},
xticklabels={10,11,12,13,14,15,16,17,18,19,20},
xticklabels={10,11,12,13,14,15,16,17,18,19,20},
xticklabels={10,11,12,13,14,15,16,17,18,19,20},
xticklabels={10,11,12,13,14,15,16,17,18,19,20},
y grid style={darkgray176},
ylabel={Time (s)},
ymin=-100, ymax=700,
ytick style={color=black}
]
\addplot [semithick, steelblue31119180, mark=+, mark size=3, mark options={solid}]
table {%
0 0.1063
1 0.7206
2 1.347
3 3.138
4 8.807
5 23.28
6 58.39
7 201.6
8 529.4
9 nan
10 nan
};
\addlegendentry{\Dtrv}
\addplot [semithick, darkorange25512714, mark=x, mark size=3, mark options={solid}]
table {%
0 0.1075
1 0.7484
2 1.402
3 3.189
4 9.913
5 22.58
6 59.41
7 206.1
8 532.1
9 nan
10 nan
};
\addlegendentry{\Drrs}
\addplot [semithick, forestgreen4416044, mark=*, mark size=3, mark options={solid}]
table {%
0 0.00059
1 0.000679
2 0.000672
3 0.000938
4 0.001449
5 0.001032
6 0.001425
7 0.001618
8 0.001751
9 0.001986
10 0.002191
};
\addlegendentry{\Dpu}
\end{axis}

\end{tikzpicture}
	\caption{
		Vanilla \qute (\Dtrv) vs \qute with \Drrs and \Dpu on Bridged ts-\lqparity.
		The $x$-axis gives $n$.
		\Dtrv and \Drrs timed out at 10 minutes for $n \geq 19$.
	}
	\label{fig:qute-bridged-ts-lqparity}
\end{figure}

We also evaluated \qute on the PCNF tracks of QBFEvals 2020 and 2022.
On QBFEVal 2022, there is unfortunately no independence to be detected even with \Dpu, so the experiment makes little sense.
In 15 minutes and out of 434 instances of the PCNF track of QBFEval 2022\footnote{\url{https://www.qbflib.org/qbfeval2022_results.php}~\cite{QBFEval2006,PulinaS19}}, vanilla \qute solved 74, with \Drrs 72, and with \Dpu 70 instances, in line with the expectation that adding dependency schemes makes the solver worse because the extra effort for computing independence is fruitless (the set of instances solved by \Drrs is a susbet of those solved by vanilla, and those solved by \Dpu are a subset of those solved by \Drrs).
On QBFEval 2020 the outcome is as follows.
In 900 seconds (and out of 521 instances in total) vanilla \qute solved 103 formulas (76 UNSAT) with PAR2%
\footnote{Penalized average runtime, i.e. average runtime where unsolved instances count for $900\times 2$ seconds.}
score $1462.89$, \qute+\Drrs solved 92 formulas (65 UNSAT) with PAR2 score $1511.97$, and \qute+\Dpu solved 90 formuals (62 UNSAT) with PAR2 score $1519.97$.
There are two instances solved by \qute+\Dpu and not solved by vanilla \qute:
\begin{verbatim}
	s38584_PR_9_5.qdimacs
	trivial_query71_1344.qdimacs
\end{verbatim}
Five instances solved by \qute+\Dpu and not \qute+\Drrs:
\begin{verbatim}
	biu.mv.xl_ao.bb-b003-p020-IPF03-c03.blif-biu.inv.prop.bb-bmc.conf07.01X-QBF\
	.BB1-Zi.BB2-Zi.BB3-Zi.with-IOC.unfold-008.qdimacs
	biu.mv.xl_ao.bb-b003-p020-MIF02-c01.blif-biu.inv.prop.bb-bmc.conf06.01X-QBF\
	.BB1-01X.BB2-Zi.BB3-Zi.with-IOC.unfold-010.qdimacs
	biu.mv.xl_ao.bb-b003-p020-MIF02-c01.blif-biu.inv.prop.bb-bmc.conf07.01X-QBF\
	.BB1-Zi.BB2-Zi.BB3-Zi.with-IOC.unfold-009.qdimacs
	biu.mv.xl_ao.bb-b003-p020-MIF04-c05.blif-biu.inv.prop.bb-bmc.conf03.01X-QBF\
	.BB1-Zi.BB2-Zi.BB3-01X.with-IOC.unfold-010.qdimacs
	s38584_PR_9_5.qdimacs
\end{verbatim}
Seven instances solved by \qute+\Drrs and not by \qute+\Dpu:
\begin{verbatim}
	fpu-10Xh-correct04-uniform-depth-18.qdimacs
	incrementer-enc07-nonuniform-depth-25.qdimacs
	incrementer-enc07-uniform-depth-25.qdimacs
	szymanski-20-s.qdimacs
	szymanski-24-s.qdimacs
	tlc05-uniform-depth-70.qdimacs
	tlc05-uniform-depth-80.qdimacs
\end{verbatim}
Twelve instances solved by vanilla \qute but not with \Drrs:
\begin{verbatim}
	axquery_query71_1344.qdimacs
	biu.mv.xl_ao.bb-b003-p020-IPF02-c05.blif-biu.inv.prop.bb-bmc.conf06.01X-QBF\
	.BB1-01X.BB2-Zi.BB3-Zi.with-IOC.unfold-009.qdimacs
	biu.mv.xl_ao.bb-b003-p020-IPF02-c05.blif-biu.inv.prop.bb-bmc.with-IOC.unfold-010.qdimacs
	biu.mv.xl_ao.bb-b003-p020-IPF03-c03.blif-biu.inv.prop.bb-bmc.conf07.01X-QBF\
	.BB1-Zi.BB2-Zi.BB3-Zi.with-IOC.unfold-008.qdimacs
	biu.mv.xl_ao.bb-b003-p020-MIF02-c01.blif-biu.inv.prop.bb-bmc.conf06.01X-QBF\
	.BB1-01X.BB2-Zi.BB3-Zi.with-IOC.unfold-010.qdimacs
	biu.mv.xl_ao.bb-b003-p020-MIF02-c01.blif-biu.inv.prop.bb-bmc.conf07.01X-QBF\
	.BB1-Zi.BB2-Zi.BB3-Zi.with-IOC.unfold-009.qdimacs
	biu.mv.xl_ao.bb-b003-p020-MIF04-c05.blif-biu.inv.prop.bb-bmc.conf03.01X-QBF\
	.BB1-Zi.BB2-Zi.BB3-01X.with-IOC.unfold-010.qdimacs
	filesys_smbmrx_cvsndrcv.c.qdimacs
	fpu-01Xh-error02-uniform-depth-24.qdimacs
	fpu-10Xh-error01-uniform-depth-20.qdimacs
	fpu-10Xh-error01-uniform-depth-25.qdimacs
	tlc05-uniform-depth-85.qdimacs
\end{verbatim}
And fifteen instances solved by vanilla \qute but not with \Dpu:
\begin{verbatim}
	axquery_query71_1344.qdimacs
	biu.mv.xl_ao.bb-b003-p020-IPF02-c05.blif-biu.inv.prop.bb-bmc.conf06.01X-QBF\
	.BB1-01X.BB2-Zi.BB3-Zi.with-IOC.unfold-009.qdimacs
	biu.mv.xl_ao.bb-b003-p020-IPF02-c05.blif-biu.inv.prop.bb-bmc.with-IOC.unfold-010.qdimacs
	filesys_smbmrx_cvsndrcv.c.qdimacs
	fpu-01Xh-error02-uniform-depth-24.qdimacs
	fpu-10Xh-correct04-uniform-depth-18.qdimacs
	fpu-10Xh-error01-uniform-depth-20.qdimacs
	fpu-10Xh-error01-uniform-depth-25.qdimacs
	incrementer-enc07-nonuniform-depth-25.qdimacs
	incrementer-enc07-uniform-depth-25.qdimacs
	szymanski-20-s.qdimacs
	szymanski-24-s.qdimacs
	tlc05-uniform-depth-70.qdimacs
	tlc05-uniform-depth-80.qdimacs
	tlc05-uniform-depth-85.qdimacs
\end{verbatim}

Interestingly, on many of the instances solved by vanilla \qute but not by \qute with \Drrs or \Dpu, the computation of the dependency scheme tends to take up a large fraction, sometimes all, of the time.
This is more pronounced for \Drrs than for \Dpu: on instances solved by vanilla but not \Drrs the dependency scheme computation almost always takes up more than 50\% of the time.
For \Dpu there are some instances where it is fast, and some where it takes all the time.
This shows that \qute's strategy of computing only parts of the dependency scheme on demand is correct in principle, and in fact the solver should probably be even more strict and limit dependency computation attempts when they have consumed too large a fraction of time so far.

\Dpu at least shows the promise in that it detects more independence faster compared to \Drrs, with the caveat that the implementation in \qute does not always compute the dependencies of the same set of variables (\qute does not compute the whole dependency scheme upfront, but instead queries only parts of it dynamically to save effort; which parts will be queried depends on the search path of the solver and is highly unpredictable).

More experiments will be necessary to determine whether there are practical instances on which \Dpu (or even \Drrs for that matter) provides a significant boost.

\end{document}